\documentclass{article}

\usepackage[utf8]{inputenc}

\usepackage[margin=1in,a4paper]{geometry}
\usepackage{amsmath,amsthm,amssymb}
\usepackage{mathtools}
\usepackage{thmtools}
\usepackage{enumitem}
\usepackage{hyperref}
\usepackage{authblk}
\usepackage{orcidlink}

\setlist[itemize]{noitemsep, nolistsep,topsep=0pt}
\setlist[enumerate]{noitemsep, nolistsep,topsep=0pt}

\usepackage{graphicx}
\usepackage[footnotesize]{caption}
\usepackage{subcaption}
\usepackage{float}
\usepackage{xcolor}

\usepackage{microtype}

\theoremstyle{plain}
\newtheorem{theorem}{Theorem}[section]
\newtheorem{lemma}[theorem]{Lemma}
\newtheorem{proposition}[theorem]{Proposition}
\newtheorem{corollary}[theorem]{Corollary}
\newtheorem{observation}[theorem]{Observation}

\theoremstyle{definition}
\newtheorem{definition}[theorem]{Definition}

\theoremstyle{remark}
\newtheorem{remark}[theorem]{Remark}

\usepackage[capitalise, nameinlink]{cleveref}

\crefname{section}{Section}{Sections}
\crefname{figure}{Figure}{Figures}
\crefname{table}{Table}{Tables}
\crefname{lemma}{Lemma}{Lemmas}
\crefname{theorem}{Theorem}{Theorems}
\crefname{corollary}{Corollary}{Corollaries}
\crefname{observation}{Observation}{Observations}
\crefname{proposition}{Proposition}{Propositions}
\crefname{conjecture}{Conjecture}{Conjectures}
\crefname{example}{Example}{Examples}
\crefname{definition}{Definition}{Definitions}
\crefname{remark}{Remark}{Remarks}
\crefname{algocf}{Algorithm}{Algorithms}

\usepackage{tikz}
\usetikzlibrary{arrows.meta}
\tikzstyle{node}=[circle, draw=black, fill=black, scale=0.25]
\tikzstyle{edge}=[]
\tikzstyle{dashed_edge}=[dashed]
\tikzstyle{arc}=[-{Latex[scale=1.2]}]
\tikzstyle{dashed_arc}=[-{Latex[scale=1.2]}, dashed]
\tikzstyle{none}=[]
\pgfdeclarelayer{nodelayer}
\pgfdeclarelayer{edgelayer}
\pgfsetlayers{main,edgelayer,nodelayer}

\usepackage[linesnumbered, ruled, vlined]{algorithm2e}
\SetKwInput{Input}{Input} 
\SetKwInput{Output}{Output}
\SetKwProg{Init}{initialize}{}{}
\SetKwProg{Proc}{procedure}{}{}
\SetKw{KwInit}{initialize}
\SetKw{KwOr}{or}
\SetKwProg{ElseIf}{else if}{ then}{}
\SetKwProg{ForEach}{for each}{ do}{}
\SetKwProg{Ifif}{if}{ then}{}

\SetCommentSty{mycommfont}

\newcommand{\Oh}{\mathcal{O}}
\newcommand{\N}{N}
\newcommand{\T}{T}

\newcommand{\X}{\mathcal{X}}

\newcommand{\Sub}{\mathcal{S}}
\newcommand{\Y}{\mathcal{Y}}
\newcommand{\Z}{\mathcal{Z}}
\newcommand{\gen}{\mathrm{gen}}

\newcommand{\class}{\mathcal{E}}

\makeatletter
\def\titlefootnote{\xdef\@thefnmark{}\@footnotetext}
\makeatother

\makeatletter
\renewcommand{\maketitle}{
  \bgroup
  \setlength{\parindent}{0pt}
  \begin{flushleft}
    {\huge \textbf{\@title}}  
    \vskip 1em 
    {\large \@author}
  \end{flushleft}
  \egroup
}
\makeatother

\renewenvironment{abstract}
 {\par\noindent\textbf{\abstractname.}\ \ignorespaces}
 {\par\medskip}

\newenvironment{keywords}
 {\par\noindent\textbf{Keywords.}\ \ignorespaces}
 {\par\medskip}

\usepackage[backend=biber,
            giveninits=true,
            maxbibnames=12,
            style=numeric,
            sortcites]{biblatex}
\title{Encoding level-3 semi-directed phylogenetic networks by quarnets and quinnets}

\author[1,*]{Niels Holtgrefe}
\author[2]{Katharina T. Huber}
\author[1]{Leo van Iersel}
\author[2]{Vincent Moulton}
\affil[1]{Delft Institute of Applied Mathematics, Delft University of Technology, 
Delft, The Netherlands}
\affil[2]{School of Computing Sciences, University of East Anglia, 
Norwich, United Kingdom}
\affil[*]{Corresponding author: n.a.l.holtgrefe@tudelft.nl}
\begin{document}

\maketitle

\begin{abstract}
Phylogenetic networks generalize phylogenetic trees as models of evolutionary history, allowing lineages to merge as well as to diverge. For many types of genetic data the root position of such a network cannot be recovered, so that only a \emph{semi-directed} network
can be inferred: a mixed graph in which only the edges entering a reticulation vertex are directed. A common strategy for inferring such a network is to first infer the subnetwork it induces on each set of $k\geq 3$ of its leaves, called a \emph{$k$-net}, and then to assemble these pieces. This can only succeed if the $k$-nets determine the network, in which case that network is said to be \emph{encoded} by its $k$-nets. Semi-directed networks of level-1 and~2, those whose biconnected components contain at most one, respectively two, reticulations, are known to be encoded by their $4$-nets, or \emph{quarnets}, whereas level-3 networks are not. Even so, in this paper we show that level-3 semi-directed networks are encoded by their $5$-nets, or \emph{quinnets}, and we characterize the limitation of quarnets exactly: we show that a single previously reported counterexample captures the only obstruction, every other level-3 network being encoded by its quarnets. Our proofs rest on a collection of encoding results for individual structural features of a network, which we establish for networks of arbitrary level and which are of independent interest.
\end{abstract}

\begin{keywords}
    Phylogenetic network, Semi-directed network, Quarnet, Quinnet, Encoding, Network generator, Level-3 network.
\end{keywords}

\section{Introduction}

Evolutionary histories of collections of species are not always tree-like. Where isolation may cause lineages to diverge, reticulate processes such as hybridization, horizontal gene transfer, and recombination cause lineages to merge. \emph{Phylogenetic networks} have become the standard combinatorial model in evolutionary biology for representing non-tree-like evolutionary histories~\cite{Huson2010,Steel2016,Kong2022}. In their rooted form, these networks are a type of directed acyclic graph in which the direction of every edge records the flow of time and the leaves represent a set of extant taxa that correspond to the species of interest.
Under several  models of DNA evolution, however, the location of the root is unidentifiable: two evolutionary histories differing only in their root position may give rise to the exact same observed data~\cite{Gross2018,Xu2022,SolisLemus2016}. In such situations, all one can hope to reliably infer is a \emph{semi-directed} (phylogenetic) network~\cite{SolisLemus2016}, that is, a phylogenetic network in which all edges but the incoming edges of reticulations are undirected (see~\cref{fig:phylozoo_example,fig:preliminaries}). Such a semi-directed network may then be rooted if necessary using an outgroup. Consequently, a growing number of software tools return semi-directed networks instead of rooted networks~\cite{SolisLemus2026}, and it is 
this type of network that concerns us here. In addition, for the sake of simplicity, we will consider only binary networks, i.e., each internal vertex is either a bifurcation (one lineage splits into two) or a bireticulation (two lineages combine into one).

\begin{figure}[htb]
    \centering
    \includegraphics[width=0.86\textwidth]{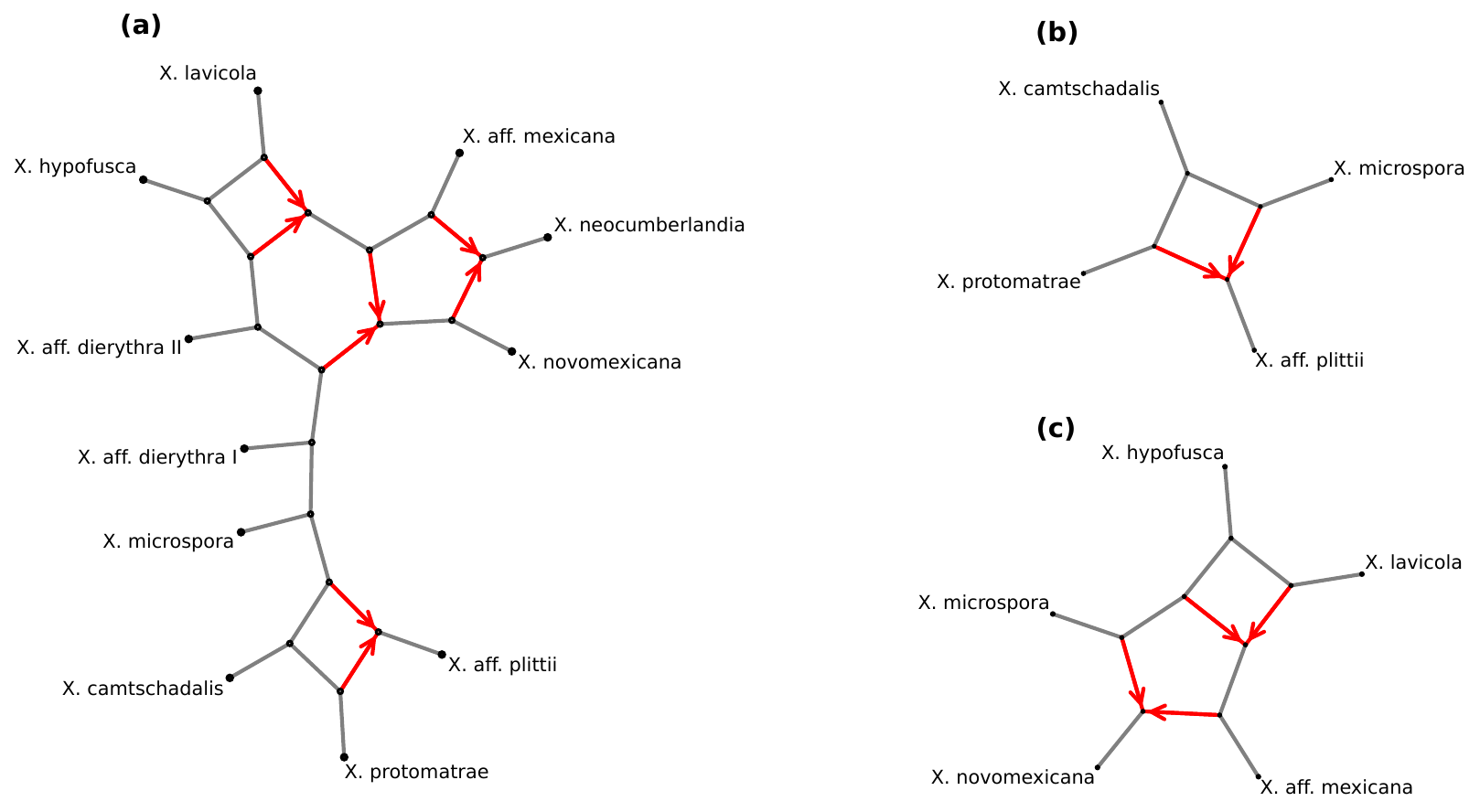}
    \caption{(a): A semi-directed phylogenetic network~$N$ on~11
    leaves, with~4
    reticulations and of level~3. The network was obtained from the rooted network
    for the
    Holarctic clade of the fungal genus \emph{Xanthoparmelia} from Keuler
    et al.~\cite{keuler2022} by suppressing the outgroup and then passing to the
    semi-directed network. Each clade is represented by a single taxon from that clade. (b): One
    quarnet of~$N$, which is of level~1. (c): One quinnet of~$N$, which is of level~2. The directed edges entering reticulations are marked in red; all other edges are undirected. Drawings were produced with \textsf{PhyloZoo}~\cite{holtgrefe2026phylozoo}.}
    \label{fig:phylozoo_example}
\end{figure}

One common approach to inferring
a (rooted or semi-directed) network is to avoid searching the space of networks on all taxa directly, and instead to infer a network (or a set of trees) on each small subset of taxa~\cite{Martin2026,Barton2026} and then to assemble these into a single network on all taxa~\cite{Holtgrefe2025,Allman2025,Willson2026,Oldman2016,Huber2011}. 
This design is inherited from the tree setting, where it rests on a classical fact: an unrooted phylogenetic tree is determined by its set of quartets---the trees it induces on four-element subsets of its leaf set~\cite{Colonius1981}. Whether generalized versions of this statement hold for networks is a combinatorial question that precedes any algorithmic considerations. Indeed, unless its small subnetworks uniquely determine a network, no algorithm can be guaranteed to reconstruct it uniquely from them. 

We call a $k$-leaf subnetwork induced by a semi-directed network~$\N$ on $n$ leaves, for some~$3\leq k\leq n$, a \emph{$k$-net} of~$\N$, with 4-nets and 5-nets called \emph{quarnets} and \emph{quinnets}, respectively  (see \cref{sec:prelims} for formal definitions and \cref{fig:phylozoo_example} for examples).
A network~$N$ is said to be \emph{encoded} by its $k$-nets if no network that is not isomorphic to~$N$ has the exact same set of $k$-nets. The general question 
that is of interest to us in this paper is the following: which semi-directed networks are encoded by their $k$-nets, and for which~$k$? This is a difficult problem in general, and so much 
of the work that has been done so far on finding a solution has focused
on networks that have low complexity in that they are quite close to being 
trees. More formally, a binary network is \emph{level-$k$} if every biconnected component of its underlying undirected graph contains at most $k$~reticulations. 
In particular, the level-0 networks are precisely the phylogenetic 
trees, and level-1 networks are those in which all cycles are vertex disjoint.

For rooted networks, subnetwork encoding results are well understood for 
level at most 3, and also for special classes such as so-called 
tree-child and orchard networks~\cite{Huber2013,vanIersel2014,Semple2021,nipius2020rooted}. However, these results do not transfer directly to semi-directed networks, and only one result is known in that setting: Huber et al.~\cite{huber2025encoding} showed
that binary, level-2 (and hence level-1), semi-directed networks are encoded by their quarnets. For level-3, however, the picture changes: \cref{fig:counterexample} (generalized from Figure~2 in \cite{huber2025encoding}) exhibits two non-isomorphic level-3 semi-directed networks with identical quarnets, showing that quarnets no longer suffice to encode all level-$k$ networks beyond level-$2$. Consequently, even if one were able to infer the complete set of quarnets perfectly from data, no reconstruction algorithm based solely on assembling quarnets 
would be able to uniquely recover the underlying network. This leaves open the problem of whether subnetworks on a larger number of leaves suffice to encode level-3 networks.

In this paper, we resolve this question and show that all binary, level-3, semi-directed networks are encoded by their quinnets (\cref{thm:main-encoding}(b)). Additionally, we show that the structure in \cref{fig:counterexample} is the only obstruction to quarnet encoding: Quarnets encode a level-3 network exactly up to an explicitly described equivalence arising from that example (\cref{thm:main-encoding}(a)).

Our proofs rely on a collection of encoding results for structural features of a network, and on a case-by-case analysis of the semi-directed level-3 \emph{generators}, small multi-graphs from which every biconnected component of a network is obtained by attaching leaves along its edges~\cite{gambette2009structure} (see \cref{fig:generators}). Roughly speaking, these features describe which generator a biconnected component has (\cref{cor:generator_encoding}), along which edge of that generator each leaf is attached (\cref{lem:equivalence_class_encoding}), and in which order the leaves along an edge appear (\cref{lem:ordering}), and we show that each of them is determined by small subnetworks. These encoding results are established for networks of arbitrary level, and may thus be of independent interest. Together they considerably reduce the case analysis that remains for level-3. Indeed, it suffices to consider generators without parallel edges (\cref{lem:parallel_contraction}), and for most of these the encoding of the features already determines the network, leaving only a small number of cases to be treated separately.

The remainder of this paper is organized as follows. \Cref{sec:prelims} contains preliminaries on mixed graphs, phylogenetic networks, and their generators. \Cref{sec:encoding_features} develops the results on the encoding of certain structural features for networks of arbitrary level. In \cref{sec:encoding_level3} these results are combined with a case analysis to prove our main result for level-3 networks~(\cref{thm:main-encoding}). We conclude with a discussion in \cref{sec:discussion}.

\section{Preliminaries}\label{sec:prelims}

\subsection{Mixed graphs}
A \emph{mixed graph} is a graph that may contain both undirected and directed edges. A \emph{mixed multi-graph} may additionally contain parallel edges. In this work, we only consider mixed multi-graphs where all parallel edges appear as pairs of parallel directed edges oriented in the same direction. A \emph{directed multi-graph} is a mixed multi-graph without undirected edges. For ease of notation, we will use $uv$ to denote either a directed edge from~$u$ to~$v$ or an undirected edge between~$u$ and~$v$.

Suppose that $G$ is a mixed multi-graph. The \emph{indegree} of a vertex is the number of incoming directed edges, its \emph{outdegree} is the number of outgoing directed edges, and its \emph{degree} is the total number of incident edges (all accounting for multiplicities). A \emph{leaf} is a degree-1 vertex and a \emph{reticulation vertex} is a vertex with indegree at least~2. The incoming edges of a reticulation vertex are called \emph{reticulation edges}. Let $v$ be a degree-2 vertex whose indegree and outdegree are both not equal to~2 and with incident edges $uv$ and $vw$. By \emph{suppressing}~$v$, we mean removing $v$ and adding the edge~$e = uw$, where~$e$ is directed towards~$w$ if $uv$ or $vw$ is directed, and $e$ is undirected if both $uv$ and $vw$ are undirected. Conversely, by \emph{subdividing} an edge $e = uw$ we mean removing $e$, adding a \emph{subdivision vertex}~$v$, and adding edges $e_1 = uv$ and $e_2 = vw$, where $e_1$ is undirected, $e_2$ is directed towards~$w$ if $e$ was directed towards~$w$, and $e_2$ is undirected otherwise. 

A \emph{path}~$P$ between two vertices $v_1$ and $v_p$ of $G$ is a sequence of pairwise distinct vertices~$(v_1, \ldots, v_p)$, $p \geq 1$, such that, in case $p\geq 2$, we have for all~$i \in \{1, \ldots, p-1\}$, that $v_i$ and $v_{i+1}$ are adjacent (and may thus be connected by a directed or an undirected edge). We say that the edges between $v_i$ and $v_{i+1}$, $i \in \{1, \ldots, p-1\}$, are \emph{contained} in~$P$.
We call $P$ a \emph{semi-directed path} (from~$v_1$ to~$v_p$) if, in addition, 
each edge $v_iv_{i+1}$ contained in~$P$ is undirected or directed towards~$v_{i+1}$. Furthermore, we call $P$ an \emph{up-down path} if there exists some $j \in \{1, \ldots, p\}$ such that~$(v_j, \ldots, v_1)$ and $(v_j, \ldots, v_p)$ are semi-directed paths. Note that semi-directed paths are also up-down paths.

A \emph{cycle}~$C$ of $G$ is a sequence of pairwise distinct vertices~$(v_1, \ldots, v_p, v_{p+1} = v_1)$, $p \geq 3$, such that~$(v_1, \ldots, v_p)$ is a path and $v_1$ and $v_p$ are adjacent. The edges \emph{contained} in~$C$ are the edges contained in the path~$(v_1, \ldots, v_p)$ and the edges $(v_1,v_p)$. We call~$C$ a \emph{semi-directed cycle} if~$(v_1, \ldots, v_p)$ is a semi-directed path from $v_1$ to $v_p$ and the edges $(v_1,v_p)$ are undirected or directed towards~$v_1$.

We call~$G$ \emph{connected} if there exists a path between each pair of vertices. Otherwise we call $G$ \emph{disconnected}. In case $G$ is connected, a \emph{cut-edge} of $G$ is an edge whose removal results in a disconnected graph. A \emph{blob} in $G$ is a maximal connected subgraph of $G$ without any cut-edges and an edge is \emph{incident} to a blob~$B$ if exactly one of its incident vertices is a vertex in $B$.
An \emph{$m$-blob}, $m \ge 1$, is a blob in $G$ that is incident to $m$~edges. Thus, each leaf of $G$ is a 1-blob, whereas a 2-blob of $G$ can be a single degree-2 vertex, or a more complex structure. By \emph{contracting} an $m$-blob, $m\geq 2$, we mean identifying all the vertices in the blob (and removing self-loops). By \emph{suppressing} a 2-blob, we mean contracting the 2-blob and additionally suppressing the resulting degree-2 vertex if it does not have indegree 2 or outdegree 2.

An \emph{isomorphism} from a mixed multi-graph~$G$ with leaf set~$\X$ to another mixed multi-graph~$H$ on the same leaf set is a bijection between their vertex sets which is the identity on $\X$ such that an edge between two vertices in $G$ exists if and only if an edge between the corresponding vertices in $H$ exists and such that, for each edge in~$G$, the corresponding edge in~$H$ has (i) the same type (directed or undirected); (ii) the same orientation in the directed case; and (iii) the same multiplicity. In addition, $G$ is \emph{isomorphic} to $H$, denoted by $G \cong H$, if there exists an isomorphism from~$G$ to $H$. An isomorphism from a mixed multi-graph to itself is an \emph{automorphism}.

\subsection{Phylogenetic networks}

Throughout, we assume~$\X$ to be a set of at least two elements and we let $n = |\X|$. In the context of phylogenetics these elements represent \emph{taxa}: the taxonomic units under consideration.

\begin{definition}[Rooted network]\label{def:rooted_network}A \emph{(binary) rooted (phylogenetic) network} on~$\X$ is a rooted, directed acyclic graph~$\N^+$ such that (i) the root has indegree zero and outdegree two; (ii) the leaves have indegree one, outdegree zero and they are bijectively labeled by $\X$; (iii) all other vertices either have indegree one and outdegree two, or indegree two and outdegree one. 
\end{definition}

Given a rooted network on~$\X$, we will use~$\X$ to also refer to its leaf set, since the leaves of the network and the set of taxa~$\X$ are in bijection.
Throughout this work, we make the common assumption that all rooted networks on~$\X$ are \emph{LSA networks}~\cite{fischer2010new}: rooted networks in which the root is the \emph{lowest stable ancestor} of $\X$, i.e., it is the lowest vertex through which all directed paths from the root to any leaf in $\X$ must pass. See \cref{fig:preliminaries}a for an example rooted network.

\begin{figure}[htb]
    \centering
    \pgfdeclarelayer{nodelayer}
\pgfdeclarelayer{edgelayer}
\pgfsetlayers{main,edgelayer,nodelayer}
\tikzset{font=\footnotesize}

\tikzstyle{none}=[]
\tikzstyle{node}=[circle, draw=black]
\tikzstyle{leaf_node}=[circle, draw=black, fill=black, scale=0.275] 
\tikzstyle{internal_node}=[circle, draw=black, fill=white, scale=0.2]

\tikzstyle{mini_label}=[none,scale=1]
\tikzstyle{small_label}=[scale=1.15]
\tikzstyle{medium_label}=[none,scale=1.0]
\tikzstyle{main_label}=[none,scale=1.1]

\tikzstyle{edge}=[]
\tikzstyle{dashed_edge}=[densely dotted]
\tikzstyle{special_edge}=[densely dotted]
\tikzstyle{special_edge2}=[]
\tikzstyle{ret_arc}=[-{Latex[scale=.9]}]
\tikzstyle{dashed_arc}=[-{Latex[scale=1.2]}, dashed]
\tikzstyle{arc}=[-{Latex[scale=.9]}]

\begin{tikzpicture}[scale=0.4]
	\begin{pgfonlayer}{nodelayer}
		\node [style={internal_node}] (130) at (-3.75, 32.25) {};
		\node [style={internal_node}] (131) at (-1.25, 31.5) {};
		\node [style={internal_node}] (132) at (-3.75, 29.75) {};
		\node [style={internal_node}] (133) at (-1.5, 29) {};
		\node [style={internal_node}] (134) at (-3.25, 28.5) {};
		\node [style={internal_node}] (135) at (-2.25, 27.5) {};
		\node [style={leaf_node}] (136) at (-1, 27.75) {};
		\node [style={leaf_node}] (137) at (-2.25, 26.25) {};
		\node [style={leaf_node}] (139) at (-0.25, 32) {};
		\node [style={internal_node}] (140) at (-4.25, 27.75) {};
		\node [style={internal_node}] (141) at (-4, 26.25) {};
		\node [style={internal_node}] (142) at (-5.25, 27) {};
		\node [style={leaf_node}] (143) at (-4, 25.25) {};
		\node [style={leaf_node}] (144) at (-5.75, 26) {};
		\node [style={leaf_node}] (145) at (-0.25, 29.25) {};
		\node [style={mini_label}] (147) at (-5, 30) {$x_1$};
		\node [style={mini_label}] (148) at (-0.25, 28.75) {$x_3$};
		\node [style={mini_label}] (149) at (-1, 27.25) {$x_4$};
		\node [style={mini_label}] (150) at (-2.25, 25.75) {$x_5$};
		\node [style={mini_label}] (151) at (-4, 24.75) {$x_6$};
		\node [style={mini_label}] (152) at (-5.75, 25.5) {$x_7$};
		\node [style={internal_node}] (155) at (8.75, 31) {};
		\node [style={internal_node}] (156) at (7.25, 30) {};
		\node [style={internal_node}] (157) at (10.25, 30) {};
		\node [style={internal_node}] (158) at (7.25, 28.5) {};
		\node [style={internal_node}] (159) at (8.75, 27.5) {};
		\node [style={leaf_node}] (160) at (11.5, 30.25) {};
		\node [style={leaf_node}] (161) at (8.75, 26.5) {};
		\node [style={internal_node}] (162) at (8, 32.25) {};
		\node [style={leaf_node}] (163) at (5.25, 31.5) {};
		\node [style={internal_node}] (164) at (6.25, 27.75) {};
		\node [style={internal_node}] (165) at (6, 26.25) {};
		\node [style={internal_node}] (166) at (4.75, 27.25) {};
		\node [style={leaf_node}] (167) at (6.25, 25.25) {};
		\node [style={leaf_node}] (168) at (3.75, 26.5) {};
		\node [style={leaf_node}] (169) at (11.5, 28) {};
		\node [style={mini_label}] (171) at (5, 31) {$x_1$};
		\node [style={mini_label}] (172) at (11.5, 29.75) {$x_3$};
		\node [style={mini_label}] (173) at (11.5, 27.5) {$x_4$};
		\node [style={mini_label}] (174) at (8.75, 26) {$x_5$};
		\node [style={mini_label}] (175) at (6.25, 24.75) {$x_6$};
		\node [style={mini_label}] (176) at (3.75, 26) {$x_7$};
		\node [style={internal_node}] (177) at (-0.75, 30.25) {};
		\node [style={internal_node}] (178) at (10.25, 28.5) {};
		\node [style={small_label}] (179) at (10.5, 31.5) {$\N$};
		\node [style={small_label}] (180) at (1, 33.25) {$\N^+$};
		\node [style={internal_node}] (202) at (29.5, 30.25) {};
		\node [style={internal_node}] (203) at (28.25, 31) {};
		\node [style={internal_node}] (204) at (30.25, 30) {};
		\node [style={leaf_node}] (207) at (31.25, 30.5) {};
		\node [style={internal_node}] (209) at (29.5, 31.75) {};
		\node [style={leaf_node}] (210) at (30.5, 32.25) {};
		\node [style={internal_node}] (211) at (26.75, 31) {};
		\node [style={internal_node}] (212) at (25.5, 30.25) {};
		\node [style={internal_node}] (213) at (25.5, 31.75) {};
		\node [style={leaf_node}] (214) at (24.5, 29.75) {};
		\node [style={leaf_node}] (215) at (24.5, 32.25) {};
		\node [style={leaf_node}] (216) at (31.25, 29.5) {};
		\node [style={mini_label}] (217) at (31, 32.25) {$x_1$};
		\node [style={mini_label}] (218) at (31.75, 30.5) {$x_3$};
		\node [style={mini_label}] (219) at (31.75, 29.5) {$x_4$};
		\node [style={mini_label}] (221) at (24, 29.75) {$x_6$};
		\node [style={mini_label}] (222) at (24, 32.25) {$x_7$};
		\node [style={internal_node}] (225) at (28.5, 27.5) {};
		\node [style={internal_node}] (226) at (26.75, 27.5) {};
		\node [style={internal_node}] (227) at (26.75, 25.75) {};
		\node [style={leaf_node}] (228) at (29.25, 28.25) {};
		\node [style={leaf_node}] (229) at (26, 25) {};
		\node [style={leaf_node}] (235) at (26, 28.25) {};
		\node [style={leaf_node}] (237) at (29.25, 25) {};
		\node [style={mini_label}] (239) at (29.75, 28.25) {$x_3$};
		\node [style={mini_label}] (240) at (29.75, 25) {$x_4$};
		\node [style={mini_label}] (241) at (25.5, 25) {$x_5$};
		\node [style={mini_label}] (242) at (25.5, 28.25) {$x_7$};
		\node [style={internal_node}] (244) at (28.5, 25.75) {};
		\node [style={main_label}] (245) at (-2.5, 23.75) {$(a)$};
		\node [style={main_label}] (246) at (8.5, 23.75) {$(b)$};
		\node [style={main_label}] (247) at (27.75, 23.75) {$(d)$};
		\node [style={internal_node}] (248) at (-4, 31) {};
		\node [style={leaf_node}] (249) at (-4.75, 30.5) {};
		\node [style={leaf_node}] (251) at (10.25, 33) {};
		\node [style={mini_label}] (252) at (-0.25, 31.5) {$x_2$};
		\node [style={mini_label}] (253) at (10.85, 33) {$x_2$};
		\node [style={internal_node}] (254) at (-2.25, 34) {};
		\node [style={internal_node}] (255) at (-1.5, 32.75) {};
		\node [style={internal_node}] (256) at (9.5, 32.25) {};
		\node [style={internal_node}] (257) at (6.5, 31.25) {};
		\node [style={leaf_node}] (263) at (19, 31) {};
		\node [style={leaf_node}] (264) at (18.75, 28.5) {};
		\node [style={leaf_node}] (266) at (16.5, 30.25) {};
		\node [style={internal_node}] (269) at (16.75, 28) {};
		\node [style={leaf_node}] (270) at (16.75, 26.75) {};
		\node [style={leaf_node}] (271) at (15.5, 27.5) {};
		\node [style={leaf_node}] (272) at (19.75, 29.75) {};
		\node [style={mini_label}] (273) at (16.25, 29.75) {$x_1$};
		\node [style={mini_label}] (274) at (19.5, 31.25) {$x_3$};
		\node [style={mini_label}] (275) at (19.75, 29.25) {$x_4$};
		\node [style={mini_label}] (276) at (18.75, 28) {$x_5$};
		\node [style={mini_label}] (277) at (16.75, 26.25) {$x_6$};
		\node [style={mini_label}] (278) at (15.5, 27) {$x_7$};
		\node [style={small_label}] (280) at (20.5, 32.5) {$\T$};
		\node [style={leaf_node}] (281) at (17.5, 31.5) {};
		\node [style={mini_label}] (282) at (18, 31.75) {$x_2$};
		\node [style={internal_node}] (283) at (18, 29.75) {};
		\node [style={main_label}] (284) at (17.5, 23.75) {$(c)$};
	\end{pgfonlayer}
	\begin{pgfonlayer}{edgelayer}
		\draw [style=arc] (130) to (131);
		\draw [style=arc] (132) to (134);
		\draw [style=arc] (134) to (135);
		\draw [style=arc] (133) to (135);
		\draw [style=arc] (135) to (137);
		\draw [style=arc] (133) to (136);
		\draw [style=arc] (134) to (140);
		\draw [style=arc] (140) to (141);
		\draw [style=arc] (140) to (142);
		\draw [style=arc] (142) to (141);
		\draw [style=arc] (142) to (144);
		\draw [style=arc] (141) to (143);
		\draw [style=edge] (157) to (155);
		\draw [style=edge] (157) to (160);
		\draw [style=edge] (156) to (158);
		\draw [style=edge] (159) to (161);
		\draw [style=edge] (158) to (164);
		\draw [style=edge] (164) to (166);
		\draw [style=edge] (166) to (168);
		\draw [style=edge] (165) to (167);
		\draw [style={ret_arc}] (158) to (159);
		\draw [style={ret_arc}] (164) to (165);
		\draw [style={ret_arc}] (166) to (165);
		\draw [style=arc] (177) to (145);
		\draw [style=arc] (177) to (133);
		\draw [style=arc] (131) to (177);
		\draw [style={ret_arc}] (178) to (159);
		\draw [style=edge] (169) to (178);
		\draw [style=edge] (178) to (157);
		\draw [style=edge] (202) to (209);
		\draw [style=edge] (204) to (202);
		\draw [style=edge] (209) to (210);
		\draw [style=edge] (204) to (207);
		\draw [style=edge] (211) to (213);
		\draw [style=edge] (213) to (215);
		\draw [style=edge] (212) to (214);
		\draw [style={ret_arc}] (211) to (212);
		\draw [style={ret_arc}] (213) to (212);
		\draw [style={ret_arc}] (209) to (203);
		\draw [style={ret_arc}] (202) to (203);
		\draw [style=edge] (204) to (216);
		\draw [style=edge] (203) to (211);
		\draw [style=edge] (225) to (228);
		\draw [style=edge] (227) to (229);
		\draw [style={ret_arc}] (226) to (227);
		\draw [style={ret_arc}] (244) to (227);
		\draw [style=edge] (237) to (244);
		\draw [style=edge] (244) to (225);
		\draw [style=edge] (226) to (235);
		\draw [style=edge] (225) to (226);
		\draw [style=arc] (130) to (248);
		\draw [style=arc] (248) to (132);
		\draw [style=arc] (248) to (249);
		\draw [style=arc] (254) to (255);
		\draw [style=arc] (255) to (131);
		\draw [style=arc] (254) to (130);
		\draw [style=arc] (255) to (139);
		\draw [style=arc] (131) to (132);
		\draw [style={ret_arc}] (256) to (155);
		\draw [style={ret_arc}] (162) to (155);
		\draw [style=edge] (162) to (256);
		\draw [style=edge] (256) to (251);
		\draw [style={ret_arc}] (155) to (156);
		\draw [style={ret_arc}] (257) to (156);
		\draw [style=edge] (162) to (257);
		\draw [style=edge] (257) to (163);
		\draw [style=edge] (269) to (271);
		\draw [style=edge] (270) to (269);
		\draw [style=edge] (269) to (283);
		\draw [style=edge] (283) to (264);
		\draw [style=edge] (272) to (283);
		\draw [style=edge] (283) to (263);
		\draw [style=edge] (283) to (281);
		\draw [style=edge] (283) to (266);
	\end{pgfonlayer}
\end{tikzpicture}
    \caption{(a): A rooted network~$\N^{+}$ on $\X = \{x_1, \ldots, x_7\}$ of level 3. (b): The level-3, semi-directed network~$\N$ on $\X$ obtained from $\N^{+}$. (c):~The tree-of-blobs~$\T$ of~$\N$. (d): The subnetworks of $\N$ induced by $\{x_1, x_3, x_4, x_6, x_7\}$ (top) and $\{x_3, x_4, x_5, x_7\}$ (bottom).}
    \label{fig:preliminaries}
\end{figure}
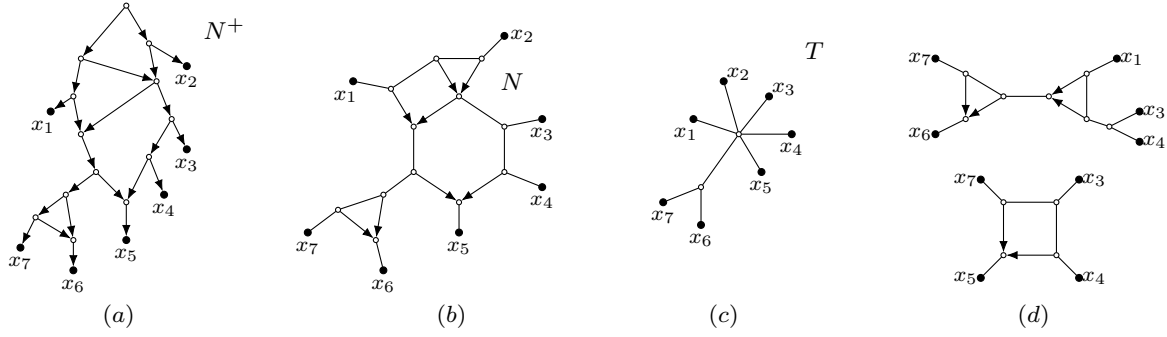

\begin{definition}[Semi-directed network]\label{def:semi_directed}
A \emph{(binary) semi-directed (phylogenetic) network} on~$\X$ is a mixed graph~$\N$ that can be obtained from a rooted LSA network~$\N^+$ on~$\X$ by undirecting all non-reticulation edges and, exhaustively, identifying parallel edges, suppressing degree-2 vertices, and suppressing 2-blobs until this is no longer possible.
\end{definition}

We again use~$\X$ to refer to the leaf set of a semi-directed network on $\X$. We note that the notion of a semi-directed network is well-defined, since the order of the operations in \cref{def:semi_directed} does not change the resulting network, as shown in~\cite{huber2025encoding}. Also note that in some definitions of semi-directed networks, parallel edges and 2-blobs are preserved (see \cref{sec:discussion} for a discussion on how our results translate to that setting). See \cref{fig:preliminaries}b for an example of a semi-directed network. 

A (semi-directed or rooted) network is called (i) \emph{simple}, if it contains at most one $m$-blob, $m \geq 3$; (ii) \emph{level-$k$}, $k\geq 0$, if there are at most~$k$ reticulation vertices in each blob of the network; and (iii) \emph{strictly level-$k$}, $k \ge 1$, if it is level-$k$ and not level-$(k-1)$.
Observe that both the rooted network~$\N^+$ in \cref{fig:preliminaries}a and the semi-directed network~$\N$ in \cref{fig:preliminaries}b are strictly level-3. 
Neither of the networks is simple.
A semi-directed network $\N$ on $\X$ without reticulation vertices is an \emph{(unrooted) phylogenetic tree} (so that, in particular, all
of the edges in the tree are undirected). The \emph{tree-of-blobs} of $N$ is the (possibly non-binary) phylogenetic tree on $\X$ obtained by contracting every $m$-blob, $m\geq 3$, to a single vertex (see \cref{fig:preliminaries}c).

In \cref{def:restriction} we formalize the notion of a \emph{subnetwork}, which is again well-defined (see~\cite{huber2025encoding}).
\begin{definition}[Subnetwork]\label{def:restriction}
Given a semi-directed network $\N$ on $\X$ and some $\Y \subseteq \X$ with $| \Y | \geq 2$, the \emph{subnetwork} of $\N$ induced by $\Y$ is the semi-directed network $\N|_\Y$ on $\Y$ obtained from $\N$ by taking the union of all up-down paths between leaves in $\Y$, followed by exhaustively identifying parallel edges, suppressing degree-2 vertices, and suppressing 2-blobs until this is no longer possible.
\end{definition}
A 4-leaf subnetwork induced by an $n$-leaf semi-directed network $\N$ is called a \emph{quarnet} (of~$\N$) and a 5-leaf subnetwork a \emph{quinnet} (see \cref{fig:preliminaries}d);
in general we call a $k$-leaf subnetwork a \emph{$k$-net} (of~$\N$).
We denote the set of pairwise non-isomorphic $k$-nets of~$\N$ by $\Sub_k (\N)$, where $3 \leq k \leq n$. By abuse of notation, we also put $\Sub_k (\N) = \Sub_n (\N)$ for $k \geq n$ and $\Sub_k (\N) = \Sub_{ \min \{3, n\}} (\N)$ for $k \in \{0, 1, 2\}$.

For two sets of semi-directed networks $\mathcal{S}, \mathcal{S}'$, we write $\mathcal{S} \simeq \mathcal{S}'$ if there exists a bijection~$\phi: \mathcal{S} \to \mathcal{S}'$ such that $\N \cong \phi (\N)$ for all $\N \in \mathcal{S}$.
We say that a semi-directed network~$N$ is \emph{encoded} by $k$-nets, $k\geq 2$, if for every binary semi-directed network~$\N'$ with $\Sub_k (\N) \simeq \Sub_k (\N')$ it holds that~$\N\cong \N'$. We say that a class~$\mathcal{C}$ of networks is \emph{encoded} by $k$-nets if each network in~$\mathcal{C}$ is encoded by $k$-nets.

\subsection{Generators}\label{subsec:prelim_generator}
Given a simple rooted network~$\N^+$, the \emph{rooted generator}~$\gen (\N^+ )$ is the directed multi-graph that can be obtained from~$\N^+$ by removing all leaves and their incident edges and suppressing all resulting vertices having both indegree~1 and outdegree~1. Similarly, the \emph{semi-directed generator}~$\gen (\N )$ of a simple semi-directed network~$\N$ is the mixed multi-graph obtained from $\N$ by removing all leaves and their incident edges and suppressing all resulting vertices of degree-2 except those with indegree-2. See \cref{fig:level2_generators} for examples. For technical reasons, we define the generator of a simple, strictly level-1, semi-directed network to be a single vertex with a bidirected self-loop. Note that a rooted generator must necessarily have a unique root.

\begin{definition}[Generator]\label{def:generator}
    A mixed (resp. directed) multi-graph~$G$ is a \emph{level-$k$, semi-directed} (resp. \emph{rooted}) \emph{generator} for some integer $k\geq 1$ if $G$ is isomorphic to the generator of a simple, strictly level-$k$, semi-directed (resp. rooted) network.
\end{definition}
We often omit `rooted' or `semi-directed' when referring to a generator if the meaning is clear from context. Observe that the generators in \cref{fig:level2_generators} are level-2 generators. Moreover, \cite{huber2025encoding} showed that the two generators in \cref{fig:level2_generators}d are \emph{all} of the possible semi-directed level-2 generators.

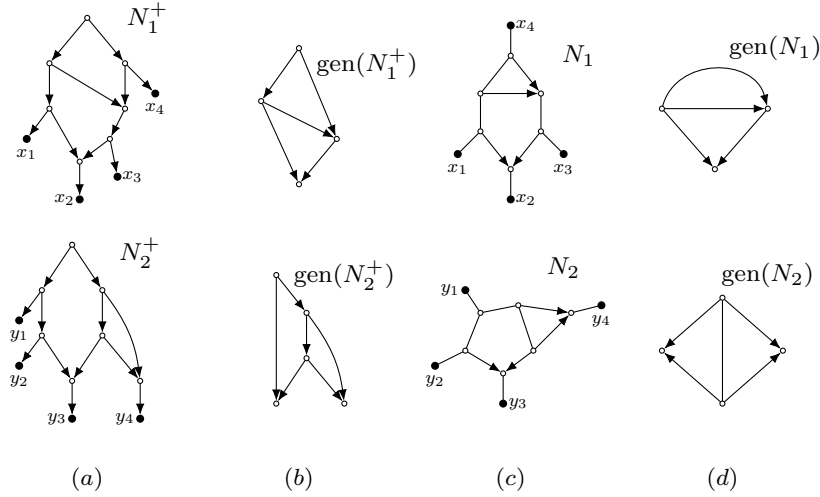
\begin{figure}[htb]
    \centering
    \pgfdeclarelayer{nodelayer}
\pgfdeclarelayer{edgelayer}
\pgfsetlayers{main,edgelayer,nodelayer}
\tikzset{font=\footnotesize}

\tikzstyle{none}=[]
\tikzstyle{node}=[circle, draw=black]
\tikzstyle{leaf_node}=[circle, draw=black, fill=black, scale=0.275] 
\tikzstyle{internal_node}=[circle, draw=black, fill=white, scale=0.2]

\tikzstyle{mini_label}=[none,scale=1.15]
\tikzstyle{small_label}=[scale=0.85]
\tikzstyle{medium_label}=[none,scale=1.1]
\tikzstyle{main_label}=[none,scale=1.3]

\tikzstyle{edge}=[]
\tikzstyle{dashed_edge}=[densely dotted]
\tikzstyle{special_edge}=[densely dotted]
\tikzstyle{special_edge2}=[]
\tikzstyle{ret_arc}=[-{Latex[scale=.9]}]
\tikzstyle{dashed_arc}=[-{Latex[scale=1.2]}, dashed]
\tikzstyle{arc}=[-{Latex[scale=.9]}]

\begin{tikzpicture}[scale=0.4]
	\begin{pgfonlayer}{nodelayer}
		\node [style={internal_node}] (12) at (32, 29) {};
		\node [style={internal_node}] (13) at (33.75, 27) {};
		\node [style={internal_node}] (14) at (35.5, 29) {};
		\node [style={internal_node}] (20) at (34, 22.75) {};
		\node [style={internal_node}] (21) at (32, 21) {};
		\node [style={internal_node}] (22) at (34, 19.25) {};
		\node [style={internal_node}] (28) at (36, 21) {};
		\node [style={internal_node}] (29) at (13, 32) {};
		\node [style={internal_node}] (30) at (11.75, 30.5) {};
		\node [style={internal_node}] (31) at (14.25, 30.5) {};
		\node [style={leaf_node}] (32) at (15.25, 29.5) {};
		\node [style={internal_node}] (33) at (14.25, 29) {};
		\node [style={internal_node}] (34) at (11.75, 29) {};
		\node [style={internal_node}] (35) at (12.75, 27.25) {};
		\node [style={leaf_node}] (36) at (11, 28) {};
		\node [style={leaf_node}] (37) at (12.75, 26) {};
		\node [style={leaf_node}] (38) at (14, 26.75) {};
		\node [style={internal_node}] (39) at (13.75, 28) {};
		\node [style={internal_node}] (40) at (13.5, 23) {};
		\node [style={internal_node}] (41) at (13.5, 21.5) {};
		\node [style={internal_node}] (42) at (12.5, 20) {};
		\node [style={internal_node}] (43) at (14.75, 20) {};
		\node [style={internal_node}] (44) at (12.5, 24.5) {};
		\node [style={leaf_node}] (45) at (12.5, 18.75) {};
		\node [style={leaf_node}] (46) at (14.75, 18.75) {};
		\node [style={leaf_node}] (47) at (10.75, 22) {};
		\node [style={internal_node}] (48) at (11.5, 23) {};
		\node [style={leaf_node}] (50) at (10.75, 20.5) {};
		\node [style={internal_node}] (51) at (11.5, 21.5) {};
		\node [style={internal_node}] (52) at (20, 31) {};
		\node [style={internal_node}] (53) at (18.75, 29.25) {};
		\node [style={internal_node}] (56) at (21.25, 28) {};
		\node [style={internal_node}] (58) at (20, 26.5) {};
		\node [style={internal_node}] (59) at (20.25, 22.25) {};
		\node [style={internal_node}] (60) at (20.25, 20.75) {};
		\node [style={internal_node}] (61) at (19.25, 19.25) {};
		\node [style={internal_node}] (62) at (21.5, 19.25) {};
		\node [style={internal_node}] (63) at (19.25, 23.5) {};
		\node [style={internal_node}] (76) at (26, 29.5) {};
		\node [style={internal_node}] (77) at (27, 30.75) {};
		\node [style={leaf_node}] (78) at (27, 31.75) {};
		\node [style={internal_node}] (79) at (28, 29.5) {};
		\node [style={internal_node}] (80) at (26, 28.25) {};
		\node [style={internal_node}] (81) at (27, 27) {};
		\node [style={leaf_node}] (82) at (25.25, 27.5) {};
		\node [style={leaf_node}] (83) at (27, 26) {};
		\node [style={leaf_node}] (84) at (28.75, 27.5) {};
		\node [style={internal_node}] (85) at (28, 28.25) {};
		\node [style={small_label}] (86) at (11, 27.5) {$x_1$};
		\node [style={small_label}] (87) at (12.25, 26) {$x_2$};
		\node [style={small_label}] (88) at (14.5, 26.75) {$x_3$};
		\node [style={small_label}] (89) at (15.25, 29) {$x_4$};
		\node [style={small_label}] (90) at (27.5, 31.75) {$x_4$};
		\node [style={small_label}] (91) at (25.25, 27) {$x_1$};
		\node [style={small_label}] (92) at (27.5, 26) {$x_2$};
		\node [style={small_label}] (93) at (28.75, 27) {$x_3$};
		\node [style={small_label}] (94) at (10.75, 21.5) {$y_1$};
		\node [style={small_label}] (95) at (10.75, 20) {$y_2$};
		\node [style={small_label}] (96) at (12, 18.75) {$y_3$};
		\node [style={small_label}] (97) at (14.25, 18.75) {$y_4$};
		\node [style={medium_label}] (102) at (13, 16.75) {$(a)$};
		\node [style={medium_label}] (103) at (20, 16.75) {$(b)$};
		\node [style={medium_label}] (104) at (27, 16.75) {$(c)$};
		\node [style={medium_label}] (105) at (34, 16.75) {$(d)$};
		\node [style={mini_label}] (106) at (15, 32) {$\N^{+}_1$};
		\node [style={mini_label}] (107) at (14.75, 24.25) {$\N^{+}_2$};
		\node [style={mini_label}] (108) at (22.25, 30.5) {$\gen ( \N^{+}_1 )$};
		\node [style={mini_label}] (109) at (21.5, 23.5) {$\gen ( \N^{+}_2 )$};
		\node [style={mini_label}] (110) at (29.25, 30.75) {$\N_1$};
		\node [style={mini_label}] (112) at (35.75, 31) {$\gen (\N_1)$};
		\node [style={mini_label}] (113) at (35.5, 23.5) {$ \gen (\N_2 )$};
		\node [style={internal_node}] (114) at (27.25, 22.5) {};
		\node [style={internal_node}] (115) at (27.75, 21) {};
		\node [style={internal_node}] (116) at (26.75, 20.25) {};
		\node [style={internal_node}] (117) at (29, 22.25) {};
		\node [style={leaf_node}] (118) at (26.75, 19.25) {};
		\node [style={leaf_node}] (119) at (30, 22.5) {};
		\node [style={leaf_node}] (120) at (25.5, 23) {};
		\node [style={internal_node}] (121) at (26, 22.25) {};
		\node [style={leaf_node}] (122) at (24.5, 20.5) {};
		\node [style={internal_node}] (123) at (25.5, 21) {};
		\node [style={small_label}] (124) at (27.25, 19.25) {$y_3$};
		\node [style={small_label}] (125) at (24.5, 20) {$y_2$};
		\node [style={small_label}] (126) at (25, 23) {$y_1$};
		\node [style={small_label}] (127) at (30, 22) {$y_4$};
		\node [style={mini_label}] (128) at (28.75, 23.75) {$\N_2$};
	\end{pgfonlayer}
	\begin{pgfonlayer}{edgelayer}
		\draw [style={ret_arc}] (12) to (14);
		\draw [style={ret_arc}] (14) to (13);
		\draw [style={ret_arc}] (12) to (13);
		\draw [style={ret_arc}, bend left=75, looseness=1.25] (12) to (14);
		\draw [style={ret_arc}] (22) to (21);
		\draw [style={ret_arc}] (20) to (21);
		\draw [style={ret_arc}] (20) to (28);
		\draw [style={ret_arc}] (22) to (28);
		\draw [style=edge] (20) to (22);
		\draw [style=arc] (29) to (30);
		\draw [style=arc] (29) to (31);
		\draw [style=arc] (30) to (33);
		\draw [style=arc] (31) to (33);
		\draw [style=arc] (31) to (32);
		\draw [style=arc] (30) to (34);
		\draw [style=arc] (34) to (35);
		\draw [style=arc] (35) to (37);
		\draw [style=arc] (34) to (36);
		\draw [style=arc] (33) to (39);
		\draw [style=arc] (39) to (38);
		\draw [style=arc] (39) to (35);
		\draw [style=arc] (41) to (42);
		\draw [style=arc] (40) to (41);
		\draw [style=arc, bend left=15] (40) to (43);
		\draw [style=arc] (44) to (40);
		\draw [style=arc] (44) to (48);
		\draw [style=arc] (48) to (47);
		\draw [style=arc] (42) to (45);
		\draw [style=arc] (43) to (46);
		\draw [style=arc] (41) to (43);
		\draw [style=arc] (48) to (51);
		\draw [style=arc] (51) to (50);
		\draw [style=arc] (51) to (42);
		\draw [style=arc] (52) to (53);
		\draw [style=arc] (53) to (56);
		\draw [style=arc] (53) to (58);
		\draw [style=arc] (56) to (58);
		\draw [style=arc] (52) to (56);
		\draw [style=arc] (60) to (61);
		\draw [style=arc] (59) to (60);
		\draw [style=arc, bend left=15] (59) to (62);
		\draw [style=arc] (63) to (59);
		\draw [style=arc] (60) to (62);
		\draw [style=arc] (63) to (61);
		\draw [style={ret_arc}] (80) to (81);
		\draw [style={ret_arc}] (85) to (81);
		\draw [style={ret_arc}] (77) to (79);
		\draw [style={ret_arc}] (76) to (79);
		\draw [style=edge] (76) to (77);
		\draw [style=edge] (77) to (78);
		\draw [style=edge] (79) to (85);
		\draw [style=edge] (85) to (84);
		\draw [style=edge] (83) to (81);
		\draw [style=edge] (80) to (82);
		\draw [style=edge] (80) to (76);
		\draw [style={ret_arc}] (123) to (116);
		\draw [style={ret_arc}] (115) to (116);
		\draw [style={ret_arc}] (115) to (117);
		\draw [style={ret_arc}] (114) to (117);
		\draw [style=edge] (123) to (121);
		\draw [style=edge] (121) to (114);
		\draw [style=edge] (114) to (115);
		\draw [style=edge] (116) to (118);
		\draw [style=edge] (117) to (119);
		\draw [style=edge] (123) to (122);
		\draw [style=edge] (120) to (121);
	\end{pgfonlayer}
\end{tikzpicture}
    \caption{(a): Two simple, level-2, rooted networks $\N_1^{+}$ and $\N_2^{+}$. (b): The rooted generators of $\N_1^{+}$ and $\N_2^{+}$. (c): The simple, level-2, semi-directed networks $\N_1$ and $\N_2$ obtained from $\N_1^{+}$ and $\N_2^{+}$. (d): The semi-directed generators of $\N_1$ and~$\N_2$.}
    \label{fig:level2_generators}
\end{figure}

A \emph{side} of a (rooted or semi-directed) generator~$G$ is an edge of~$G$ (called an \emph{edge side}) or a degree-2 vertex of~$G$ with indegree-2 (called a \emph{reticulation side}). If an edge side is part of a pair of parallel edges, we call it a \emph{parallel edge side} and refer to the pair as a \emph{pair of parallel edge sides}. 
By definition, one can construct a simple (rooted or semi-directed) network from its generator by attaching leaves to its sides. More specifically, a leaf can be attached to a reticulation side by adding a new edge incident to the reticulation side and the leaf. Additionally, leaves can be attached to edge sides by subdividing the edge $\ell \geq 0$ times and then attaching one leaf to each of the resulting subdivision vertices (by adding a new edge incident to the subdivision vertex and the leaf).

In \cite{gambette2009structure} it was shown how to construct all rooted level-$(k+1)$ generators from the set of level-$k$, rooted generators. In addition, in \cite{nipius2022encoding}, level-$k$, semi-directed generators are defined as mixed graphs that can be obtained from level-$k$, rooted generators by \emph{semi-directing} them; more specifically, by undirecting all edges not oriented towards a reticulation, and then suppressing the unique degree-2 root. In the following proposition we prove that the two definitions 
agree. That is, we prove that one can 
obtain all level-$k$, semi-directed generators from the level-$k$, rooted generators by semi-directing them.

\begin{proposition}\label{prop:semidirecting_rooted_generators}
    A mixed multi-graph~$G$ is a level-$k$, semi-directed generator for some integer $k\geq 2$ if and only if it can be obtained by semi-directing some level-$k$, rooted generator $G^+$.
\end{proposition}
\begin{proof}
Let~$G$ be a level-$k$, semi-directed generator and let $\N$ be a simple, level-$k$, semi-directed network such that $G \cong \gen(\N)$. By definition, there also exists a rooted network~$\N^+$ from which $\N$ can be obtained by performing the operations in \cref{def:semi_directed}. In particular, we can obtain one such network $\N^+$ by subdividing a single edge in the `source component'\footnote{Informally speaking, the source component of a semi-directed network is the set of vertices and edges that can be selected as the root of the network. Since this concept of a source component is technical and not needed elsewhere in the paper, we refer to \cite{holtgrefe2026characterizing} for a formal definition.} of $\N$ to create a root~$\rho$ and then orienting all edges away from it~\cite{holtgrefe2026characterizing}. Consequently, $G$ and $G^+ = \gen (\N^+)$ have the exact same vertex set, except for the addition of an outdegree-2 vertex (which corresponds to~$\rho$) in $G^+$. Moreover, each reticulation side of $G$ corresponds to exactly one reticulation side in $G^+$. Thus, semi-directing $G^+$ yields~$G$.

Conversely, suppose that $G$ is a mixed multi-graph that can be obtained by semi-directing some level-$k$, rooted generator $G^+$. Let $\N^+$ be a rooted network with root~$\rho$ such that $\gen (\N^+) \cong G^+$. Note that we may choose $\N^+$ such that it has at least one leaf attached to every side. Consequently, when performing the operations of \cref{def:semi_directed} to obtain the corresponding semi-directed network~$\N$, no parallel edges will appear after suppressing the root~$\rho$. Hence, $\N$ only differs from $\N^+$ by the suppression of~$\rho$ and the non-reticulation edges being undirected. An argument similar to that in the previous paragraph can then be used to show  that by construction $\gen (\N) \cong G$.
\end{proof}

In Proposition 2.3 of \cite{gambette2009structure} it was shown that level-$k$, rooted generators, for $k \geq 1$, have at most $3k-1$ vertices and $4k-2$ edges. Suppressing the root removes one vertex and one edge. Consequently, \cref{prop:semidirecting_rooted_generators} immediately yields the following bounds.

\begin{corollary}\label{cor:generator_size_bound}
A level-$k$, semi-directed generator, $k\geq 1$, has at most $3k-2$ vertices and $4k-3$ edges.
\end{corollary}

Following terminology for rooted networks~\cite{vanIersel2014,nipius2020rooted}, we call a set of sides of a level-$k$ (rooted or semi-directed) generator~$G$ a \emph{set of crucial sides} if it contains all reticulation sides and at least one parallel edge side of each pair of parallel edge sides. The \emph{cruciality} of~$G$ is the size of its smallest set of crucial sides.
A \emph{crucial leaf set} of a simple (rooted or semi-directed) network~$\N$ on $\X$ is a set of leaves~$\Y \subseteq \X$ such that the sides they are attached to form a set of crucial sides. The \emph{cruciality} of $\N$ is the cruciality of its generator, or equivalently, the size of its smallest crucial leaf set. We call $\N$ \emph{$k$-crucial}, if its cruciality is at most $k$.

As an example, $\{y_3, y_4\}$ forms a smallest set of crucial leaves of the strictly level-2 network~$\N_2$ in \cref{fig:level2_generators}c (bottom). Hence, $\N_2$ is $2$-crucial. On the other hand, $\{x_2\}$ is the unique smallest set of crucial leaves of the strictly level-2 network~$\N_1$ in \cref{fig:level2_generators}c (top). So, $\N_1$ is $1$-crucial.

By noting that each pair of parallel edge sides and each reticulation side of a level-$k$ generator correspond to a unique indegree-2 vertex of the generator, which in turn corresponds uniquely to one of the $k$~reticulation vertices of a strictly level-$k$ network, we observe the following.

\begin{observation}\label{obs:level_cruciality}
Let $\N$ be a simple, level-$k$, semi-directed network. Then, the cruciality of $\N$ is at most~$k$.
\end{observation}
Note that the semi-directed network~$\N_1$ from \cref{fig:level2_generators}c shows that the cruciality of a simple, strictly level-$k$, semi-directed network can be strictly smaller than $k$. The cruciality of a simple, strictly level-1, semi-directed network is always equal to~1. The following further useful cruciality bound follows by remarking that the construction from~\cite{gambette2009structure} for building rooted generators ensures that each level-$k$, rooted generator (and thus each level-$k$, semi-directed generator), $k\geq 1$, has at least one reticulation side.

\begin{observation}\label{obs:cruciality_parallel_edges}
    Let $\N$ be a simple, semi-directed network whose generator has a pair of parallel edge sides. Then, the cruciality of $\N$ is at least~2.
\end{observation}

\section{Encoding features of level-$k$ networks by subnetworks}\label{sec:encoding_features}

In this section, we show that certain structural features of $k$-crucial (and in particular all level-$k$), semi-directed networks are encoded by subnetworks of bounded size. Our encoding results require subnetworks with at least three leaves. Recall that for $m \in \{0,1,2\}$ and a semi-directed network~$N$ on $n$ leaves, we adopted the convention $\Sub_m(\N)=\Sub_{\min\{3,n\}}(\N)$ (see~\cref{sec:prelims}), implicitly allowing additional arbitrary leaves to be added when the induced subnetwork has fewer than three leaves.

We first restate a result from \cite{huber2025encoding}, showing that the tree-of-blobs of a semi-directed network is encoded by its quarnets. Note that a similar result was shown in \cite{allman2023tree} in a slightly different setting.

\begin{proposition}[Cor.\,5.6\, \cite{huber2025encoding}]\label{prop:tree_of_blobs_encoding}
    Let $\N_1$ and $\N_2$ be two semi-directed networks on~$\X$ and let $\T_1$ and $\T_2$ be their respective trees-of-blobs. If $\Sub_4 (\N_1 ) \simeq \Sub_4 (\N_2)$, then $\T_1 \cong \T_2$.
\end{proposition}

We now prove an important fact concerning crucial leaves (\cref{prop:crucial_leaves}). See \cite[Lemma~1]{vanIersel2014} for a similar result but for rooted networks and restricted to the case that $|\Y| = 3$. As a consequence we show in \cref{cor:generator_encoding} that if two simple, $k$-crucial, semi-directed networks have the same sets of $k$-nets then they have the same generators.

\begin{lemma}\label{prop:crucial_leaves}
Let $\N$ be a simple, semi-directed network on $\X$ and let $\Y \subseteq \X$ with $|\Y| \geq 3$. Then, $\Y \subseteq \X$ is a set of crucial leaves of $\N$ if and only if $\N|_{\Y}$ is simple and $\gen (\N) \cong \gen (\N|_{\Y})$.
\end{lemma}
\begin{proof}
    Let $\Y$ be a set of crucial leaves of $\N$ and set $\Z = \X \setminus \Y$. We show that $\N|_{\Y}$ is simple and $\gen(\N) \cong \gen(\N|_{\Y})$ by repeatedly removing leaves in $\Z$. This is trivial if $\Z = \emptyset$, so take any $x \in \Z$. Since $x \notin \Y$, it is neither attached to a reticulation side nor the unique leaf attached to a pair of parallel edge sides. Thus, constructing $\N|_{\X \setminus \{x\}}$ from $\N$ amounts solely to deleting the leaf~$x$ and suppressing a degree-2 vertex, keeping the network simple and leaving the generator unchanged. Moreover, the leaves in $\Y$ form a set of crucial leaves of $\N|_{\X \setminus \{x\}}$. Repeating this procedure for every leaf in $\Z$ yields $\N|_{\Y}$ while keeping the network simple and preserving the generator throughout. Hence, $\N|_{\Y}$ is simple and $\gen(\N) \cong \gen(\N|_{\Y})$.

    Conversely, suppose that~$\Y$ is not a crucial leaf set of $\N$. Then, there is either a reticulation side or a pair of parallel edge sides with no attached leaf in~$\Y$. In the first case, it is clear that $\N|_{\Y}$ has at least one reticulation vertex fewer than~$\N$. Since we identify parallel edges in the construction of $\N|_{\Y}$ (\cref{def:restriction}), a similar argument applies in the second case. It follows that $\N|_{\Y}$ has a strictly lower level than $\N$. Thus, either $\N|_{\Y}$ is not simple, or, if it is, $\gen (\N) \not \cong \gen (\N|_{\Y})$.
\end{proof}

\begin{proposition}\label{cor:generator_encoding}
    Let $\N_1$ and $\N_2$ be two simple, $k$-crucial, semi-directed networks on~$\X$. 
    If $\Sub_k (\N_1 ) \simeq\Sub_k (\N_2)$, then $\gen (\N_1) \cong \gen (\N_2)$.
\end{proposition}
\begin{proof}
    Let $r_1$ (resp. $r_2$) be the number of reticulations of $\N_1$ (resp. $\N_2$). Assume without loss of generality that $r_1 \geq r_2$. Let $\Y \subseteq \X$ be a smallest set of crucial leaves of $\N_1$ and note that $|\Y | \leq k$. 
    
    By \cref{prop:crucial_leaves}, $\gen (\N_1) \cong \gen (\N_1|_{\Y} )$. Since $\Sub_k (\N_1 ) \simeq\Sub_k (\N_2)$, we have $\N_1|_\Y \cong \N_2|_\Y$, and so $\gen (\N_1|_\Y) \cong \gen (\N_2|_\Y )$. Thus, $\gen (\N_1) \cong \gen (\N_2|_\Y )$ and both have $r_1$ reticulations. Observe that $\N_2|_{\Y}$ can have at most $r_2$ reticulations, since constructing a subnetwork never creates reticulations. By assumption, we also have $r_1 \geq r_2$, so $\N_2|_\Y$ and $\N_2$ have the same number of reticulations. It follows that $\Y$ is a set of crucial leaves of $\N_2$ (see the proof of \cref{prop:crucial_leaves}), and thus $\gen (\N_2|_\Y ) \cong \gen (\N_2 )$ by \cref{prop:crucial_leaves}, proving the result.
\end{proof}

We call an automorphism of a generator \emph{reticulation-preserving} if (i) it maps each reticulation side to itself; and (ii) it maps each pair of parallel edge sides to itself. We define the \emph{side-equivalence relation} on the sides of a generator~$G$ as follows: two sides are \emph{side-equivalent} if there exists a reticulation-preserving automorphism of~$G$ mapping the two sides to each other. Note that the reflexive, symmetric, and transitive properties of an equivalence relation indeed hold.
We denote by $\class ( G )$ the set of side-equivalence classes of~$G$.
See \cref{fig:equivalent_sides}a for an example: the edge sides labeled $B_1$ in generator~$G_1$ all lie in the same side-equivalence class, as $G_1$ admits several symmetries. In particular, a reticulation-preserving automorphism exists that maps the two vertices (in black) of the edge side labeled $A_1$ to each other, and similarly a reticulation-preserving automorphism exists that maps the gray vertices to one another. For generator~$G_2$, however, the edge sides labeled $B_2$ and $C_2$, as well as $D_2$ and $E_2$, are not in the same class, since any reticulation-preserving automorphism must map each pair of parallel edge sides to itself.

Given a simple, semi-directed network~$\N$ on $\X$ and a generator $G \cong \gen (N)$, the \emph{leaf-to-side function} $\sigma_\N: \X \to \class ( G )$ maps each leaf in~$\X$ to the side-equivalence class of~$G$ in which the side that it is attached to is contained. An isomorphism $\psi: \gen (\N_1) \to \gen (\N_2)$ between generators of two networks~$\N_1, \N_2$ maps side-equivalence classes to side-equivalence classes, and thus induces a bijection $\psi_*: \class(\gen(\N_1)) \to \class(\gen(\N_2))$. We call two simple, semi-directed networks $\N_1$ and $\N_2$ on~$\X$ \emph{leaf-side consistent} if there is such an isomorphism~$\psi$ with $\psi_* (\sigma_{\N_1} (x)) = \sigma_{\N_2} (x)$ for all $x \in \X$. See \cref{fig:equivalent_sides}b and~c.

Without loss of generality, we assume that, when two simple networks have isomorphic generators, they have the same generator~$G$. Moreover, when~$\N_1,\N_2$ are leaf-side consistent, we assume that $\sigma_{\N_1} (x) = \sigma_{\N_2} (x)$ for all $x \in \X$, i.e., each leaf is attached to sides from the same side-equivalence class of~$G$ in $\N_1$ and in $\N_2$.

\begin{figure}[htb]
    \centering
    \input{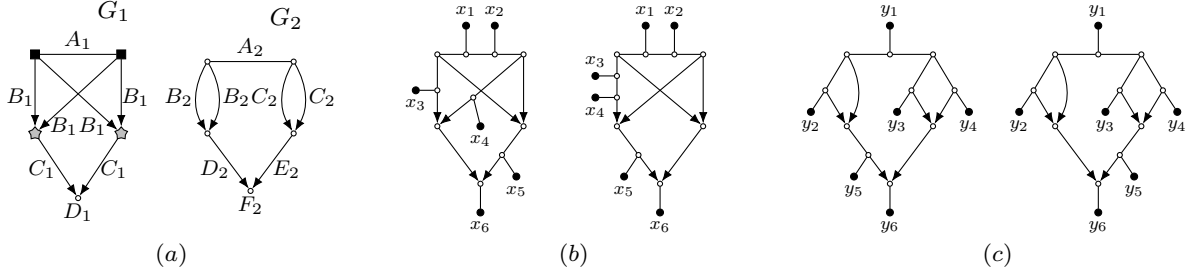}
    \caption{(a): Two semi-directed level-3 generators~$G_1$ and $G_2$. Sides with the same labels are side-equivalent (see text for further explanation). (b): Two non-isomorphic semi-directed networks on $\{x_1, \ldots, x_6\}$ that both have $G_1$ as their generator. The two networks are leaf-side consistent, since the sides to which each leaf is attached in the two networks are from the same side-equivalence classes of~$G_1$. (c): Two non-isomorphic semi-directed networks on $\{y_1, \ldots, y_6\}$ that both have $G_2$ as their generator. The two networks are not leaf-side consistent, since the sides to which $y_5$ is attached in the two networks are from different side-equivalence classes of~$G_2$.}
    \label{fig:equivalent_sides}
\end{figure}

By observing that the side-equivalence class containing the side to which a leaf~$x$ is 
attached can be determined from a subnetwork with the same generator as the original network, we obtain the following.
\begin{proposition}\label{lem:equivalence_class_encoding}
    Let $\N_1$ and $\N_2$ be two simple, $k$-crucial, semi-directed networks on~$\X$. 
    If $\Sub_{k+1} (\N_1 ) \simeq \Sub_{k+1} (\N_2)$, then $\N_1$ and $\N_2$ are leaf-side consistent.
\end{proposition}
\begin{proof}
    By \cref{cor:generator_encoding}, we may assume $\gen (\N_1) \cong \gen (\N_2)$. Towards a contradiction, assume that $\N_1$ and $\N_2$ are not leaf-side consistent. Let $x \in \X$ be such that $\sigma_{\N_1} (x) \neq \sigma_{\N_2} (x)$. Let $\Y \subseteq \X$ be a set of $k$ crucial leaves of $\N_1$ and let $\Z =\Y \cup \{ x\}$. Since $\Sub_{k+1} (\N_1 ) \simeq \Sub_{k+1} (\N_2)$, we have $\N_1|_\Z \cong \N_2|_\Z$, so $\gen (\N_1|_{\Z}) \cong \gen (\N_2|_{\Z})$. By \cref{prop:crucial_leaves}, $\gen (\N_1) \cong \gen (\N_1|_{\Z})$. Therefore, $\N_1$, $\N_2$, $\N_1|_\Z$ and $\N_2|_\Z$ all have the same generator.
    But since $\sigma_{\N_1} (x) \neq \sigma_{\N_2} (x)$, we get $\sigma_{\N_1|_{\Z}} (x) \neq \sigma_{\N_2|_{\Z}} (x)$. So, $\N_1|_{\Z} \not \cong \N_2|_{\Z}$: a contradiction.
\end{proof}

We say that a set of leaves of a simple network \emph{share a side}, if they are attached to the same side of the generator of the network. For example, in \cref{fig:equivalent_sides}b, $x_1$ and $x_2$ share a side in either of the two networks. Since pairs of parallel edge sides are always in the same side-equivalence class, two leaves in a pair of leaf-side consistent networks are either both attached to a parallel edge side, or both are not. Note that this does not rule out that two leaves are both attached to a pair of parallel edge sides, but in one network they share a parallel edge side and in the other they are attached to different parallel edge sides of the pair.
The following again uses \cref{cor:generator_encoding}, with Part~(b) of the proposition relying on the fact that a leaf attached to a parallel edge side must be part of some smallest crucial leaf set.
\begin{proposition}\label{lem:share_sides}
    Let $\N_1$ and $\N_2$ be two simple, $k$-crucial, semi-directed networks on $\X$. Let $x, y \in \X$. Suppose that one of the following holds.
    \begin{enumerate}[label=(\alph*)]
    \looseness = -1
        \item $\Sub_{k+2} (\N_1 ) \simeq \Sub_{k+2} (\N_2)$; or
        \item $\Sub_{k+1} (\N_1 ) \simeq \Sub_{k+1} (\N_2)$ and $x$ is attached to a parallel edge side in $\N_1$.
    \end{enumerate}
    Then $x$ and $y$ share a side in $\N_1$ if and only if they share a side in $\N_2$.
\end{proposition}
\begin{proof}
    First suppose that $\Sub_{k+2} (\N_1 ) \simeq \Sub_{k+2} (\N_2)$. By \cref{cor:generator_encoding} we may assume that $\N_1$ and $\N_2$ have the same generator.
    Let $\Y \subseteq \X$ be a set of $k$ crucial leaves of $\N_1$ and let $\Z = \Y \cup \{x, y\}$. Since $\Sub_{k+2} (\N_1 ) \simeq \Sub_{k+2} (\N_2)$, we have that $\N_1|_\Z \cong \N_2|_\Z$, so $\gen (\N_1|_\Z) \cong \gen (\N_2|_\Z)$. 
    By \cref{prop:crucial_leaves}, we obtain $\gen (\N_1) \cong \gen (\N_1|_\Z)$. Therefore, $\N_1$, $\N_2$, $\N_1|_\Z$ and $\N_2|_\Z$ all have isomorphic generators. It follows that $x$ and $y$ share a side in $\N_1$ if and only if they share a side in $\N_1|_\Z \cong \N_2|_\Z$ if and only if they share a side in $\N_2$. This proves the proposition in case (a) holds. If (b) holds, observe that we can choose $\Y$ such that $x \in \Y$.
\end{proof}

Before we proceed to the next encoding result (\cref{lem:ordering}), we first require a technical lemma (\cref{lem:excellent_cycle}) whose proof is based on the following results in~\cite{huber2025encoding} and \cite{holtgrefe2026characterizing}. Call a cycle~$C$ in a semi-directed network \emph{good} if there exists exactly one reticulation 
vertex~$r$ (called the \emph{sink}) in $\N$ such that $C$ contains both reticulation edges that enter~$r$. 
A good cycle is \emph{excellent} if its sink  is adjacent to a leaf. Lemma~5.4 of \cite{huber2025encoding} 
states the following, which we explicitly record here since it will be used again later.

\begin{lemma}[Lem.\,5.4~\cite{huber2025encoding}]\label{lem:excellent_original}
Let $\N$ be a simple, semi-directed network on $\X$ that is not a phylogenetic tree. Let~$e$ be an edge of~$\N$ between two non-leaf vertices. Then~$e$ is contained in at least one excellent cycle.
\end{lemma}

Furthermore, the following proof will use Theorem~2 of~\cite{holtgrefe2026characterizing}, which states that a semi-directed network cannot contain any semi-directed cycles and that a semi-directed network cannot contain paths between two distinct reticulations containing only undirected edges.

\begin{lemma}\label{lem:excellent_cycle}
    Let $\N$ be a simple, semi-directed network on $\X$ that is not a phylogenetic tree. Let $\Y \subsetneq \X$ be a set of leaves that share an edge side. Then, there exists a leaf $z \in \X \setminus \Y$ that is attached to a reticulation side such that $\N|_{\Y \cup \{z\}}$ is simple.
\end{lemma}
\begin{proof}
    Let $m = |\Y| + 1$. 
    Let $e$ be one of the edges of $\N$ that corresponds to the edge side to which the leaves in~$\Y$ are attached. Observe that $e$ is not a cut-edge. By \cref{lem:excellent_original}, it 
    follows that $e$ is contained in at least one excellent cycle~$C$. Let~$z$ be the leaf adjacent to the sink of this cycle.
    By Theorem~2 of \cite{holtgrefe2026characterizing} outlined above, (i) $C$ is not a semi-directed cycle; and (ii) no set of undirected edges contained in~$C$ forms a path in $\N$ between two reticulations. Hence, all edges in $C$ are on an up-down path between two leaves in $\Y \cup \{z\}$. Letting $V(C)$ denote the vertex set of $C$
    and observing that all leaves in $\Y \cup \{z\}$ are adjacent to distinct vertices of~$V(C)$, it follows that $\N|_{\Y \cup \{z\}}$ contains an $m$-blob. Since $\N|_{\Y \cup \{z\}}$ also has $m$ leaves, it follows that $\N|_{\Y \cup \{z\}}$ is simple.
\end{proof}

Let $S=uv$ be an edge side of a generator~$G$ of a simple, semi-directed network~$\N$ on $\X$, and let $\Y \subseteq \X$ be a set of leaves that is attached to $S$. 
The \emph{$uv$-oriented ordering} of $\Y$ induced by~$S$ is the linear ordering of the leaves in $\Y$ obtained by traversing the side from $u$ to $v$ and listing the leaves in the order in which they are attached in~$\N$ along $S$.
The \emph{unoriented ordering} of $\Y$ induced by~$S$ is this ordering considered up to reversal; equivalently, it is the $vu$-oriented ordering considered up to reversal.
We call the edge side~$S$ \emph{unoriented} if there exists a reticulation-preserving automorphism of~$G$ that maps $u$ to $v$ and $v$ to $u$. An edge side is \emph{oriented} if it is not unoriented, that is, every reticulation-preserving automorphism of~$G$ maps $u$ to $u$ and $v$ to $v$. Note that all edge sides that are directed edges are oriented, but edge sides that are undirected edges may also be oriented. 

To illustrate these definitions, observe that the side~$A_1$ of the generator $G_1$ in \cref{fig:equivalent_sides}a is unoriented, since there exists a reticulation-preserving automorphism that maps the two endpoints (in black) of that side to each other. On the other hand, the side~$A_2$ of the generator $G_2$ in \cref{fig:equivalent_sides}a, as well as all sides of the generator~$G$ from \cref{fig:ordering_sides}a, are oriented. Observe that $\Y =  \{x_1, x_2, x_3\}$ is the set of leaves of the semi-directed network~$\N$ from \cref{fig:ordering_sides}b that is attached to the edge side~$S = uv$ of the generator~$G$ from \cref{fig:ordering_sides}a. The $uv$-oriented ordering of $\Y$ induced by $S$ is $(x_1, x_2, x_3)$, with the unoriented ordering being this ordering up to reversal.

\begin{figure}[htb]
    \centering
    \input{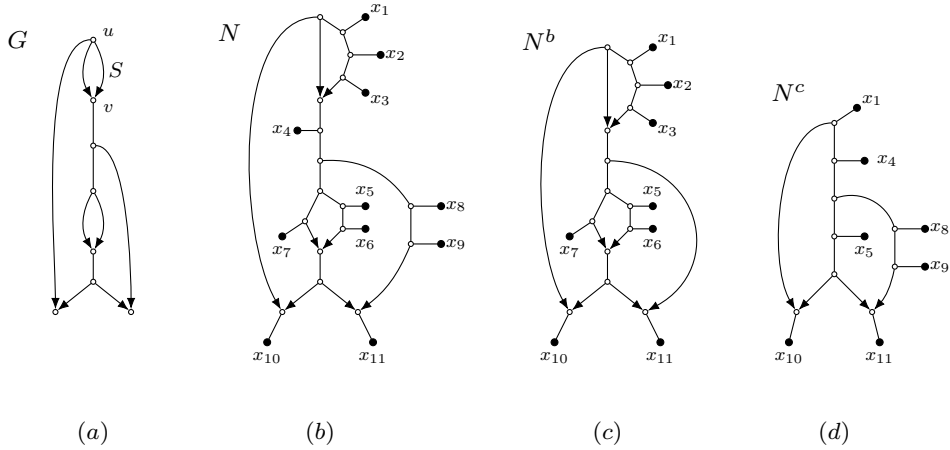}
    \caption{(a): A semi-directed level-4 generator~$G$ with a labeled edge side $S = uv$. (b): A semi-directed network~$\N$ on $\{x_1, \ldots, x_{11} \}$ with $G$ as its generator. (c): The backbone~$\N^b$ of $\N$. (d): The parallel contraction~$\N^c$ of $\N$.}
    \label{fig:ordering_sides}
\end{figure}

\begin{proposition}\label{lem:ordering}
    Let $\N_1$ and $\N_2$ be two simple, $k$-crucial, semi-directed networks on~$\X$ with the same generator~$G$. Let $\Y \subseteq \X$ be a set of leaves that all share the same edge side $S = uv$ of $G$ in both $\N_1$ and $\N_2$.
    \begin{enumerate}[label=(\alph*)]
        \item If $\Sub_{4} (\N_1) \simeq \Sub_4 (\N_2)$, then the unoriented ordering of~$\Y$ induced by~$S$ is the same in $N_1$ and $N_2$.
        \item If $\Sub_{k+2} (N_1) \simeq \Sub_{k+2} (N_2)$ and $S$ is oriented, then the $uv$-oriented ordering of~$\Y$ induced by~$S$ is the same in $N_1$ and $N_2$.
        \item If $\Sub_{k+1} (N_1) \simeq \Sub_{k+1} (N_2)$ and $S$ is a parallel edge side, then the $uv$-oriented ordering of~$\Y$ induced by~$S$ is the same in $N_1$ and $N_2$.
    \end{enumerate}
\end{proposition}
\begin{proof}
    (a): Suppose $\Sub_{4} (\N_1) \simeq \Sub_4 (\N_2)$. Observe that there is only one possible unoriented ordering of~$\Y$ if $|\Y| \leq 2$, so we may assume $|\Y | \geq 3$. Let $\Y' \subseteq \Y$ be an arbitrary subset of $\Y$ with $|\Y'| = 3$. By \cref{lem:excellent_cycle}, there is a leaf $z \in \X \setminus \Y$ that is attached to a reticulation side of~$G$ such that $\N_1|_\Z$, with $\Z =\Y' \cup \{z\}$, is simple. Since $\Sub_{4} (\N_1) \simeq \Sub_4 (\N_2)$, $\N_2|_\Z$ is also simple. Let $i \in \{1, 2\}$. Clearly, the leaves in~$\Y'$ share a side in $\N_i|_\Z$ because they shared a side in $\N_i$. Moreover, the induced unoriented ordering of $\Y'$ in $\N_i|_\Z$ is the same as the induced unoriented ordering of $\Y'$ in $\N_i$. Because $\Sub_{4} (\N_1) \simeq \Sub_4 (\N_2)$, these orderings are also the same across $\N_1$ and $\N_2$. The result follows from the fact that the unoriented ordering of~$\Y$ is fully determined by the unoriented orderings of all its size-3 subsets of leaves.

    (b): Suppose $\Sub_{k+2} (N_1) \simeq \Sub_{k+2} (N_2)$ and that $S$ is oriented. The result is trivial if $|\Y| =1$, so we assume $|\Y| \geq 2$. Let $\Y' \subseteq \Y$ be an arbitrary subset of $\Y$ with $|\Y'|=2$ and let $C \subseteq \X$ be some smallest set of crucial leaves of $\N_1$. By \cref{lem:equivalence_class_encoding}, $\N_1$ and $\N_2$ are leaf-side consistent, so $C$ is also a set of crucial leaves of $\N_2$. Let $\Z = C \cup \Y'$. Then, since $\N_i$ and $\N_i|_\Z$ have the same generator, the $uv$-oriented ordering of $\Y'$ in $\N_i|_\Z$ equals the $uv$-oriented ordering of $\Y'$ in $\N_i$, for $i \in \{1, 2\}$. Observe that $|\Z| \leq k+ 2$. So, by our assumption, $\N_1|_\Z \cong \N_2|_\Z$ and thus these $uv$-oriented orderings of $\Y'$ also coincide across $\N_1$ and $\N_2$. The result follows by noting that the $uv$-oriented ordering of~$\Y$ is fully determined by the $uv$-oriented orderings of all its size-2 subsets of leaves.

    (c): First note that parallel edge sides are directed and therefore oriented. Observe that in the proof of Part~(b), we can pick a set of crucial leaves $C \subseteq \X$ that contains a leaf from $\Y'$ if $S$ is a parallel edge side. Then, $|\Z| \leq k + 1$. The remainder of the proof is similar to the proof of Part~(b).
\end{proof}

Let $\N$ be a simple, semi-directed network on $\X$ with cruciality~$k\geq 2$. The \emph{backbone}~$\N^b$ of~$\N$ is the subnetwork $\N|_{\Y \cup \Z}$, where $\Y \subseteq \X$ is the set of leaves that are attached to parallel edge sides and $\Z \subseteq \X$ the set of leaves attached to reticulation sides (see \cref{fig:ordering_sides}c). Note that $\Y \cup \Z$ forms a set of crucial leaves of $\N$ and that $|\Y \cup \Z| \geq k \geq 2$, ensuring that the backbone is indeed a network on at least 2 leaves.

\begin{lemma}\label{lem:backbone}
    Let $\N_1$ and $\N_2$ be two simple, $k$-crucial, semi-directed networks on $\X$ with the same generator~$G$ having parallel edge sides. If $\Sub_{k+1} (\N_1 ) \simeq \Sub_{k+1} (\N_2)$, then $\N_1$ and $\N_2$ have isomorphic backbones.
\end{lemma}
\begin{proof}
Let $\N^b_i$ be the backbone of $\N_i$, $ i \in \{1, 2\}$. Note that this backbone exists, since a network whose generator contains parallel edge sides has a cruciality of at least~2 (\cref{obs:cruciality_parallel_edges}). Because the leaf set of $\N^b_i$ is a crucial leaf set of $\N_i$ by construction, it follows from \cref{prop:crucial_leaves} that $\N^b_i$ also has $G$ as its generator.
Next, note that a pair of parallel edge sides always forms its own side-equivalence class. Similarly, a reticulation side forms its own side-equivalence class.
Since~$\N_1$ and $\N_2$ are leaf-side consistent (\cref{lem:equivalence_class_encoding}), the two networks then share the same set of leaves on each pair of parallel edge sides and on each reticulation side. It follows that $\N^b_1$ and $\N^b_2$ have the same leaf sets and that they also share the same set of leaves on each pair of parallel edge sides and on each reticulation side. Moreover, for each pair of parallel edge sides, the allocation of leaves over the two sides is the same in $\N_1$ and $\N_2$ (and consequently, in $\N^b_1$ and $\N^b_2$) by \cref{lem:share_sides}b. Lastly, by \cref{lem:ordering}c, the induced $uv$-oriented ordering of the leaves attached to a parallel edge side must be the same in $\N_1$ and $\N_2$ (and consequently, in $\N^b_1$ and $\N^b_2$).  It follows that $\N^b_1 \cong \N^b_2$.
\end{proof}

Let $\N$ be a simple, semi-directed network on $\X$ with its generator containing a pair of parallel edge sides~$(S_1, S_2)$. Let $\Y \subseteq \X$ be the set of leaves attached to $S_1$ and $S_2$. To make the reduction of parallel sides canonical, we fix a total ordering of the leaves in~$\X$ and take the \emph{representative leaf} of $\Y$ to be the smallest leaf in~$\Y$ according to this ordering. In what follows, such an ordering is implicitly assumed whenever working with representative leaves, and in examples we take it to be lexicographic. By \emph{reducing} the pair of parallel edge sides~$(S_1, S_2)$ (or \emph{reducing} $\Y$) we mean the operation of contracting all edges in $\N$ that correspond to $S_1$ and $S_2$ in its generator and then removing all leaves in $\Y$, except for the representative leaf. Thus, the whole cycle is replaced by a single vertex with the representative leaf of~$\Y$ adjacent to it. The \emph{parallel contraction}~$\N^c$ of $\N$ is the semi-directed network obtained by reducing every pair of parallel edge sides. Note that the parallel contraction of $\N$ is also a semi-directed network and that its generator has no parallel edge sides. See \cref{fig:ordering_sides}d for an example.

\begin{proposition}\label{lem:parallel_contraction}
    Let $\N_1$ and $\N_2$ be two simple, $k$-crucial, semi-directed networks on $\X$ with the same generator~$G$ having parallel edge sides. If $\N_1$ and $\N_2$ are leaf-side consistent and they have isomorphic backbones, then $\Sub_{k} (\N_1 ) \simeq \Sub_{k} (\N_2)$ if and only if $\Sub_{k} (\N^c_1 ) \simeq \Sub_{k} (\N^c_2)$, where $\N^c_i$ is the parallel contraction of $\N_i$, $i \in \{1, 2\}$.
\end{proposition}
\begin{proof}
Let $\Y \subseteq \X$ be the set of leaves attached to parallel edge sides of~$\N_1$, and $\Z \subseteq \X$ the set of leaves attached to reticulation sides of~$\N_1$. Observe that then $\Y$ must also be the set of leaves attached to parallel edge sides of~$\N_2$ and $\Z$ must be the set of leaves attached to reticulation sides of~$\N_2$, since $\N_1$ and $\N_2$ are leaf-side consistent.
By \cref{prop:crucial_leaves} and since $\Y \cup \Z$ is a set of crucial leaves of both $\N_1$ and $\N_2$, the backbones $N^b_1$ and $N^b_2$ of $N_1$ and $N_2$, respectively, have~$G$ as their generator. By assumption, we also have $N^b_1 \cong N^b_2$.
It follows that $\N_1$ and $\N_2$ have the exact same allocation of leaves in $\Z$ over their reticulation sides and the same allocation of leaves in $\Y$ over their parallel edge sides. 
Moreover, since each parallel edge side~$S = uv$ of~$G$ is directed and therefore oriented, it follows from $N^b_1 \cong N^b_2$ that the $uv$-oriented ordering induced by~$S$ of the set of leaves attached to $S$ must be the same in both $\N_1$ and $\N_2$. 
This implies that any difference in $k$-nets between $\N_1$ and $\N_2$ cannot be due to differences in the $uv$-oriented ordering of the attached leaves induced by a parallel edge side.
Hence, such a difference in $k$-nets between $\N_1$ and $\N_2$ must be witnessed by differences in the $k$-nets of $\N^c_1$ and $\N^c_2$.
The equivalence stated in the proposition then readily follows.
\end{proof}

Finally, we establish a result relating the level of a network to the levels of its subnetworks, the proof of which generalizes the proof of \cite[Lemma~6.1]{huber2025encoding}. We will apply this result in the next section to distinguish level-$3$ networks from higher-level networks based on quarnets (and hence quinnets).

\begin{proposition}\label{lem:level_quarnet}
    Let~$N$ be a strictly level-$k$ semi-directed network and~$3\leq k'\leq k$. Then there exists a $k''$-net of~$N$, with~$k''\leq k'$, that is not level-$(k'-1)$.
\end{proposition}
\begin{proof}
    Let~$B$ be a strictly level-$k$ blob of~$N$. Let~$N_B$ be the simple strictly level-$k$ semi-directed network obtained from~$N$ by deleting all vertices that are not in~$B$ and not adjacent to a vertex of~$B$. By \cref{lem:excellent_original}, $N_B$ has an excellent cycle~$C_1$. Let~$x_1$ be the leaf adjacent to the sink of~$C_1$ and let~$j=1$. Repeat the following while~$|\{x_1,\ldots ,x_{j}\}|<k'$ and~$U=\cup_{i=1}^{j} C_i$ is not equal to the vertex set~$V(B)$ of~$B$. Let~$e$ be an edge incident to~$U$ (i.e., which has one endpoint in~$U$ and one endpoint not in~$U$) that is not a cut-edge. By \cref{lem:excellent_original}, $N_B$ has an excellent cycle~$C_{j+1}$ containing~$e$. Let~$x_{j+1}$ be the leaf adjacent to the sink of~$C_{j+1}$. Increase~$j$ by~$1$ and repeat. This procedure terminates when $U=\cup_{i=1}^{j} C_i=V(B)$ or~$|\{x_1,\ldots ,x_{j}\}|=k'$ or both. Let $X'=\{x_1,\ldots ,x_{j}\}$.
    
    First suppose $|X'|=k'$. Then~$U$ has~$k'$ sinks and each leaf of~$X'$ is adjacent to a sink of~$U$ in~$N_B$. Hence, the $k'$-net $N_B|_{X'}$ has at least~$k'$ reticulations. Moreover,~$U$ is biconnected by construction. It follows that $N_B|_{X'}$ is simple and not level-$(k'-1)$. Let~$X''$ be a set of~$k'$ leaves of~$N$ containing, for each leaf~$x\in X'$, a leaf~$x'$ of~$N$ that is reachable by a path from~$x$ that is disjoint from~$B$. Then $N|_{X''}$ is a $k'$-net of~$N$ that is not level-$(k'-1)$.

    Now suppose $U=V(B)$ and $|X'|<k'$. Then~$U$ contains exactly~$k$ reticulations. Suppose $N_B|_{X'}$ is strictly level~$\ell < k' \leq k$ and note that $\ell \geq |X'|$ since each leaf in $X'$ is adjacent to a sink. Then~$B$ contains $g \leq k-\ell$ good cycles that each have only two vertices with degree higher than~$2$ (so that these cycles are contracted when constructing $N_B|_{X'}$). Let~$X''$ contain all leaves of~$X'$ and, for $\max \{ g, k' - \ell \}$ good cycles, a leaf incident to that cycle in~$N_B$. Then $N_B|_{X''}$ is a $k''$-net of~$N_B$, with $k'' \leq k'$, that is not level-$(k'-1)$. It then follows, by the same argument as the final argument in the previous case, that there exists a $k''$-net of~$N$ that is not level-$(k'-1)$.
\end{proof}

\section{Encoding level-3 networks by quarnets and quinnets}
\label{sec:encoding_level3}

In this section we prove that all level-3, semi-directed networks are encoded by their quinnets, and almost all of them by their quarnets (see \cref{thm:main-encoding}). Recall that \cite{huber2025encoding} already showed that level-2, semi-directed networks are encoded by their quarnets.

\begin{proposition}[Thm.\,6.2~\cite{huber2025encoding}]\label{lem:level2_encoding}
    Let $\N_1$ and $\N_2$ be two level-2, semi-directed networks on $\X$. If $\Sub_{4} (\N_1 ) \simeq \Sub_{4} (\N_2)$, then $\N_1$ and $\N_2$ are isomorphic.
\end{proposition}

The proof of \cref{thm:main-encoding}
is based on a case analysis made on the semi-directed level-3 generators and, as we shall see, 
the results in the previous section will permit us to greatly restrict 
the number of cases that we will need to consider. For the case analysis we first need the full list of semi-directed level-3 generators, shown in \cref{fig:generators}. This list was obtained by Nipius~\cite{nipius2022encoding}, who semi-directed the $65$ rooted level-3 generators listed on \cite{gambetteGenerators} that 
had been previously computed by 
Gambette, Berry, and Paul~\cite{gambette2009structure}.\footnote{By developing new computer code, made available as part of the Python package \textsf{PhyloZoo}~\cite{holtgrefe2026phylozoo}, we 
independently checked the validity of both Nipius' and Gambette's computations.} That these mixed multi-graphs constitute \emph{all} semi-directed level-3 generators was assumed, but not proven, in~\cite{nipius2022encoding}. In light of 
\cref{prop:semidirecting_rooted_generators} it follows that this assumption is justified.

\begin{lemma}\label{lem:17level3_generators}
    There are 17 semi-directed level-3 generators and they are all depicted in \cref{fig:generators}.
\end{lemma}

\begin{figure}[htb]
    \centering
    \input{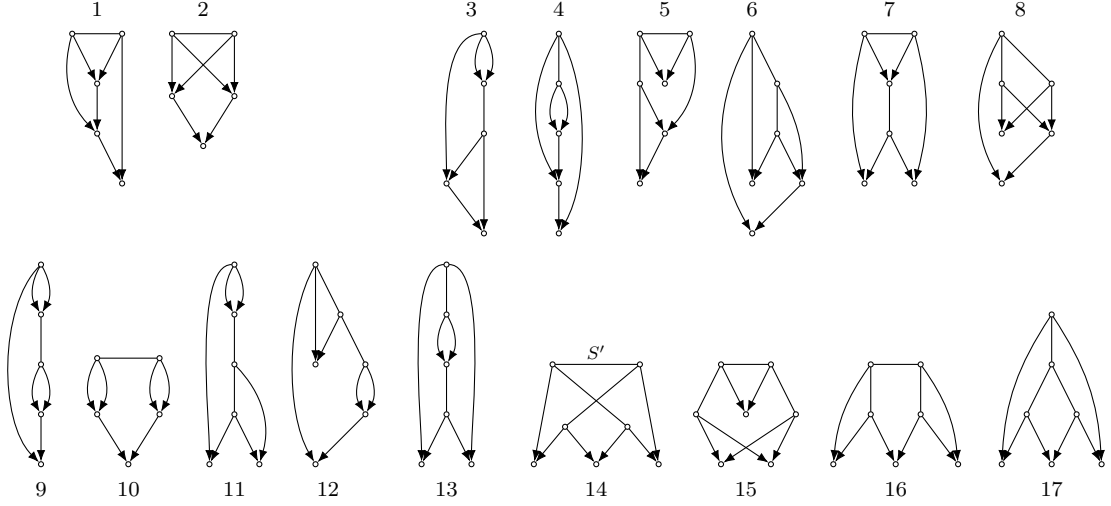}
    \caption{The 17 semi-directed level-3 generators. 
    The side~$S'$ of the generator of type~14 is referred to as its \emph{top side}.}
    \label{fig:generators}
\end{figure}

We refer to the number assigned to each of the 17 semi-directed level-3 generators 
in \cref{fig:generators}  as the \emph{type} of the generator.
Observe that generators 1 and 2 have cruciality~1, 
the generators 3--8 have cruciality~2, and the generators 9--17 have cruciality~3.
In the proof of \cref{thm:main-encoding} we  will
only need to be concerned with the generators 
that do not contain parallel edge sides, that is, 
the ones of type 1, 2, 5, 6, 7, 8, 14, 15, 16, and 17. 
Using these observations,  we now prove that
if a simple, strictly level-3 network has a generator with no parallel edge sides
that is not of type 14, then it is encoded by its quarnets. The generator of type~14 will be treated separately in \cref{lem:generator14}.

\begin{lemma}\label{lem:generator1-17}
    Let $\N_1$ and $\N_2$ be two simple, strictly level-3, semi-directed networks on $\X$ having the same generator~$G$ of type 1, 2, 5, 6, 7, 8, 15, 16, or 17. If $\Sub_{4} (\N_1 ) \simeq \Sub_{4} (\N_2)$, then $\N_1$ and $\N_2$ are isomorphic.
\end{lemma}
\begin{proof}    
    Let $k$ be the cruciality of $G$ and note that $k \leq 3$ by \cref{obs:level_cruciality}. By \cref{lem:equivalence_class_encoding}, $\N_1$ and $\N_2$ are leaf-side consistent. We consider four cases, depending on the type of~$G$. Recall that in all cases, $G$ has no parallel edge sides.
    
    \textbf{Case 1:} $G$ is of type 1, 5, 6, or 7. Observe that the following hold for each of these generator types: (i) $k \leq 2$; (ii) each side of~$G$ forms its own equivalence class; (iii) each edge side of~$G$ is oriented. Since $\N_1$ and $\N_2$ are leaf-side consistent, we may assume by (ii) that each side of~$G$ has the same set of leaves attached to it in both $\N_1$ and $\N_2$. Combining \cref{lem:ordering}b with (i) and (iii), it follows that each edge side~$S=uv$ of~$G$ induces the same $uv$-oriented ordering of the set of leaves attached to it in both $\N_1$ and $\N_2$. Thus, $\N_1 \cong \N_2$.

    \textbf{Case 2:} $G$ is of type 2 or 8. Observe that $k \leq 2$ (in particular, $k=1$ if $G$ is of type 2, and $k=2$ if $G$ is of type 8). Let $S$ be the topmost reticulation side of~$G$ according to the planar representation in \cref{fig:generators} (for type~2, this is the sole reticulation side). Since $\N_1$ and $\N_2$ are leaf-side consistent, they both have the same leaf~$x\in \X$ attached to~$S$ (Insight~2.1). Denote by $S'=ur$ and $S'' =v r$ the two oriented edge sides of $G$ incident to the reticulation side~$S$. Let $\Y'$ and $\Y''$ be the sets of leaves that are attached to $S'$ and $S''$, respectively, in $\N_1$. By \cref{lem:share_sides}, $\Y'$ and $\Y''$ are also the sets of leaves attached to $S'$ and $S''$, respectively, in $\N_2$ (Insight~2.2). By \cref{lem:ordering}b, the $ur$-oriented ordering of $\Y'$ induced by $S'$ (respectively, the $vr$-oriented ordering of $\Y''$ induced by $S''$) is the same in $\N_1$ and $\N_2$ (Insight~2.3). Lastly, for $i=1,2$, put $\N_i'\coloneq \N_i|_{X\setminus \{x\}}$. Note that $\N_i'$ is a level-2, semi-directed network. By \cref{lem:level2_encoding}, $\N'_1 \cong \N'_2$ (Insight~2.4). Together, Insights~2.1--2.4 show that $\N_1 \cong \N_2$.

    \textbf{Case 3:} $G$ is of type 15. Observe that $k=3$. Let $S_1$, $S_2$ and $S_3$ denote the reticulation sides of~$G$, ordered from top to bottom and left to right according to the planar representation in \cref{fig:generators}. Let $x, y, z \in \X$ be the leaves attached to $S_1$, $S_2$, and $S_3$, respectively, in~$\N_1$. Since $\N_1$ and $\N_2$ are leaf-side consistent, $x, y, z$ are also attached to $S_1, S_2, S_3$, respectively, in $\N_2$ (Insight~3.1). For $i \in \{1,2\}$, put $\N'_i \coloneq \N_i|_{\X \setminus \{x\}}$. Note that $\N'_i$ is a level-2, semi-directed network. By \cref{lem:level2_encoding}, $\N'_1 \cong \N'_2$ (Insight~3.2). Let $S' = ur$ and $S'' = vr$ be the two oriented edge sides of~$G$ incident to $S_1$. By Insights~3.1 and~3.2, any potential non-isomorphism between $\N_1$ and $\N_2$ can only arise from a difference in the $ur$-oriented or $vr$-oriented orderings induced by $S'$ or $S''$, respectively, of the leaves attached to them. Now put $\N''_i \coloneq \N_i|_{\X \setminus \{z\}}$, for $i \in \{1, 2\}$. Note that the oriented edge sides~$S', S''$ correspond to oriented edge sides in $\N''_i$. By \cref{lem:level2_encoding} and since $\N''_1$ and $\N''_2$ are level-2, semi-directed networks, $\N''_1 \cong \N''_2$. So, the oriented orderings induced by these two sides of the leaves attached to them coincide in $\N''_1$ and $\N''_2$, and hence also in $\N_1$ and $\N_2$. It follows that $\N_1 \cong \N_2$.

    \textbf{Case 4:} $G$ is of type 16 or 17. Observe that (i) $k = 3$; (ii) each side of $G$ forms its own equivalence class; (iii) each edge side is oriented. Since the two networks are leaf-side consistent, each side~$S$ has the same set of leaves~$\Y_S$ attached to it in both $\N_1$ and $\N_2$. So, by (ii), any potential non-isomorphism between $\N_1$ and $\N_2$ can only arise from some edge side~$S$ (which is oriented by (iii)) having a different induced oriented ordering of $\Y_S$ in $\N_1$ and $\N_2$. Note that for each such edge side~$S$, there is one leaf~$z$ attached to a reticulation side in both $\N_1$ and $\N_2$, such that the leaves in $\Y_S$ are still attached to an oriented side~$S'$ in $\N'_1 \coloneq \N_1|_{\X \setminus \{z\}}$ and $\N'_2 \coloneq \N_2|_{\X \setminus \{z\}}$. Since $\N'_1$ and $\N'_2$ are level-2, semi-directed networks, we get $\N'_1 \cong \N'_2$ by \cref{lem:level2_encoding}. Thus, the oriented orderings of $\Y_S$ induced by $S'$ coincide in $\N'_1$ and $\N'_2$. Hence, the oriented orderings of $\Y_S$ induced by~$S$ also coincide in $\N_1$ and $\N_2$. Therefore, $\N_1 \cong \N_2$.
\end{proof}

We now focus on the generator of type~14.
We refer to the side of generator 14 labeled by $S'$ in \cref{fig:generators}
as its \emph{top side}. By \emph{reversing the top side} of a simple network~$\N$ with a generator of type~14, we mean reversing the order in which the leaves are attached to its top side. In addition,
we say that two simple networks that both have a type~14 generator 
are \emph{isomorphic modulo reversing the top side} if they 
are isomorphic or they become isomorphic 
after reversing the top side for one of the networks.

The type~14 generator is the one problematic case. \Cref{fig:counterexample} exhibits, for every $m \geq 2$, a pair of non-isomorphic, simple, level-3 networks with a type~14 generator that induce exactly the same set of quarnets; they are isomorphic modulo reversing the top side, but not isomorphic. The special case $m=2$ is Figure~2 of~\cite{huber2025encoding}. The following lemma shows that this is the only obstruction: quarnets determine two simple type~14 networks up to reversing the top side, while quinnets determine them completely.

\begin{figure}[htb]
    \centering
    \input{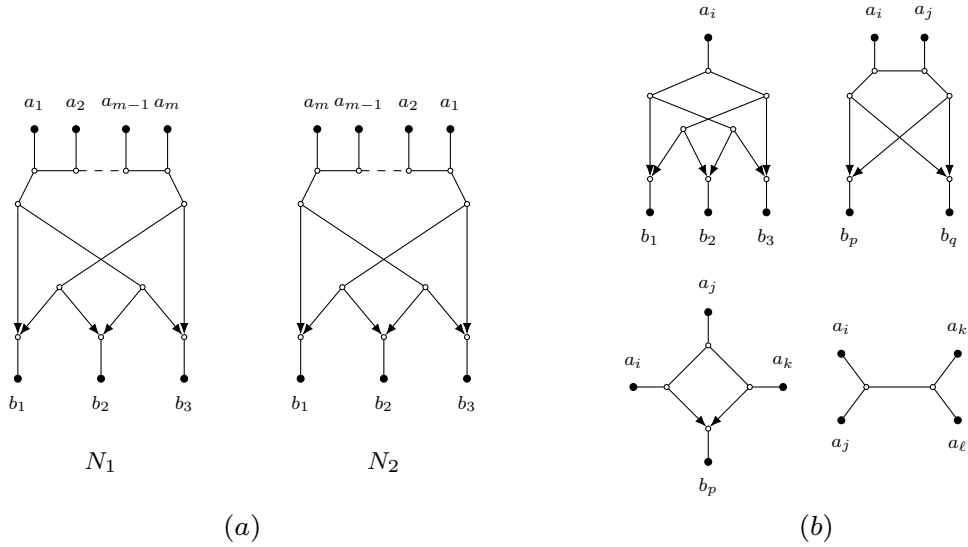}
    \caption{(a): Two level-3, semi-directed networks~$\N_1$ and~$\N_2$ on leaf set $\{a_1, \ldots, a_m, b_1, b_2, b_3\}$ ($m\geq 2$), both having the same generator of type~14. $\N_1$ and $\N_2$ are not isomorphic but are isomorphic modulo reversing the top side. (b): The types of quarnets induced by both $\N_1$ and $\N_2$, showing that $\N_1$ and $\N_2$ have the same sets of induced quarnets. The bottom left quarnet type does not appear when $m = 2$ and the bottom right quarnet type does not appear when $m \in \{2, 3\}$. The special case of this example where $m=2$ appeared as Figure~2 in~\cite{huber2025encoding}. The indices satisfy $i, j, k, \ell \in \{1, \ldots, m\}$ with $i < j < k < \ell $ and $p,q \in \{1,2,3\}$ with $p < q$.}
    \label{fig:counterexample}
\end{figure}

\begin{lemma}\label{lem:generator14}
    Let $\N_1$ and $\N_2$ be two simple, strictly level-3, semi-directed networks on $\X$ having the same generator~$G$ of type~14. 
    \begin{enumerate}[label=(\alph*)]
        \item If $\Sub_{4} (\N_1 ) \simeq \Sub_{4} (\N_2)$, then $\N_1$ and $\N_2$ are isomorphic modulo reversing the top side.
        \item If $\Sub_{5} (\N_1 ) \simeq \Sub_{5} (\N_2)$, then $\N_1$ and $\N_2$ are isomorphic.
    \end{enumerate}
\end{lemma}
\begin{proof}
    As observed at the beginning of this section, the cruciality of~$G$ is 3. Moreover, each of its sides forms its own side-equivalence class, and each edge side is oriented.
    
    (a) Suppose that $\Sub_{4} (\N_1 ) \simeq \Sub_{4} (\N_2)$. By \cref{lem:equivalence_class_encoding}, $\N_1$ and $\N_2$ are leaf-side consistent---and thus each side of~$G$ has the same leaves attached to it in both $\N_1$ and $\N_2$. By \cref{lem:ordering}a, the unoriented orderings induced by each edge side of $G$ of the leaves attached to it are the same in $\N_1$ and $\N_2$. In particular, this means that the unoriented ordering induced by the top side of~$G$ of the leaves attached to it is the same in $\N_1$ and $\N_2$. Let~$S$ be an edge side that is not the top side of~$G$ and let $\Y_S$ be the set of leaves attached to $S$ in $\N_1$ and $\N_2$. Observe that $S$ is oriented and that there is a leaf~$z \in \X$ attached to a reticulation side of~$G$ such that the 
    leaves in $\Y_S$ are still attached to an oriented side~$S'$ in $\N'_1 \coloneq \N_1|_{\X \setminus \{z\}}$ 
    and $\N'_2 \coloneq \N_2|_{\X \setminus \{z\}}$. Since $\N'_1$ and $\N'_2$ are level-2, semi-directed networks, we have $\N'_1 \cong \N'_2$ by \cref{lem:level2_encoding}. Thus, the oriented orderings of $\Y_S$ induced by $S'$ coincide in $\N'_1$ and $\N'_2$. Hence, the oriented orderings of $\Y_S$ induced by $S$ also coincide in $\N_1$ and $\N_2$. It follows that $\N_1$ and $\N_2$ are isomorphic modulo reversing the top side.

    (b) Suppose that $\Sub_{5} (\N_1 ) \simeq \Sub_{5} (\N_2)$. This implies that $\Sub_{4} (\N_1 ) \simeq \Sub_{4} (\N_2)$, so $\N_1$ and $\N_2$ are isomorphic modulo reversing the top side. By \cref{lem:ordering}b, the induced oriented orderings of the leaves attached to the top side of~$G$ are also the same for $\N_1$ and $\N_2$. Hence, $\N_1 \cong \N_2$.
\end{proof}

Putting everything together,
we now prove the main result (\cref{thm:main-encoding}), which concerns networks that need not be simple. 
To this end, we first naturally extend the notion of reversing the top side of a type~14 generator to the blobs of a non-simple network. 

Let $B$ be an $m$-blob ($m \geq 3$) of a level-3, semi-directed network~$\N$ on~$\X$. Removing $B$ partitions $\X$ into the leaf sets $X_1, \ldots, X_m$ of the remaining components; picking one leaf $x_i \in X_i$ from each yields a simple network $\N_B \coloneq \N|_{\{x_1, \ldots, x_m\}}$, whose generator~$G_B$ does not depend on this choice. We call $B$ a \emph{type~14 blob} if $G_B$ is the generator of type~14. Along the edges of~$B$ corresponding to the top side of~$G_B$ hang pendant subgraphs of $\N$ in a linear order, one for each leaf attached to the top side in the network~$\N_B$. By \emph{reversing the top side} of a type~14 blob we mean reversing the order of these pendant subgraphs, leaving the rest of $\N$ unchanged. Note that this is a slight generalization of the operation illustrated in \cref{fig:counterexample}, where each leaf on the top side may now stand for a larger pendant part of the network. Finally, two level-3 networks are \emph{isomorphic modulo reversing the top sides of type~14 blobs} if one becomes isomorphic to the other after reversing the top side of some (possibly empty) subset of its type~14 blobs.

\begin{theorem}\label{thm:main-encoding}
Let $\mathcal{L}_3$ be the class of binary, level-3, semi-directed networks. Then,
\begin{enumerate}[label=(\alph*)]
    \item $\mathcal{L}_3$ is encoded by quarnets, modulo reversing the top side of the type~14 blobs; and
    \item $\mathcal{L}_3$ is encoded by quinnets.
\end{enumerate}
\end{theorem}
\begin{proof}
    For~(a), first observe that reversing the top side of a type~14 blob leaves the quarnets unchanged. This is easily verified from \cref{fig:counterexample} for a simple network with a single type~14 blob, and the general case follows.

    Also observe that, by \cref{lem:level_quarnet}, a network not in $\mathcal{L}_3$ cannot have the same quarnets (or quinnets) as a network in $\mathcal{L}_3$. Hence, to show the claimed encoding results, it suffices to consider two 
    semi-directed networks $\N_1$ and $\N_2$ on $\X$ from the class $\mathcal{L}_3$.
    
    (a) Suppose that $\Sub_{4} (\N_1 ) \simeq \Sub_{4} (\N_2)$. Let $B_1$ be an arbitrary blob of $\N_1$. As above, observe that the removal of~$B_1$ induces a partition $\{X_1, \ldots, X_m\}$ of the leaf set $\X$ (with $m \geq 3$, $\X=\bigcup_{1 \leq i \leq m} X_i$, $X_i \subseteq \X$ and $X_i \cap X_j = \emptyset$, $1 \le  i \neq j \le m$). By \cref{prop:tree_of_blobs_encoding}, $\N_1$ and $\N_2$ have isomorphic trees-of-blobs, so there also exists a corresponding blob~$B_2$ of $\N_2$ that induces the same partition of $\X$. Pick a set of leaves~$\{x_1, \ldots, x_m \}$ such that $x_i \in X_i$ for each $1 \leq i \leq m$. Put $\N'_1 \coloneq \N_1|_{\{x_1, \ldots, x_m \}}$ and $\N'_2 \coloneq \N_2|_{\{x_1, \ldots, x_m \}}$. Since we chose an arbitrary blob~$B_1$ of $\N_1$, it now suffices to show that $\N'_1$ and $\N'_2$ are isomorphic modulo reversing the top side of the type~14 blobs.

    Observe that $\N'_1$ and $\N'_2$ are simple and that their cruciality is at most~3 (\cref{obs:level_cruciality}). By \cref{cor:generator_encoding}, they have the same generator~$G$. By \cref{lem:equivalence_class_encoding}, $\N'_1$ and $\N'_2$ are leaf-side consistent. By combining \cref{lem:backbone,lem:parallel_contraction}, we may assume that $\N'_1$ and~$\N'_2$ are their own parallel contraction or, equivalently, that~$G$ contains no parallel edge sides. If $\N'_1$ and $\N'_2$ are level-2, the result follows from \cref{lem:level2_encoding}, so we may assume that the two networks are strictly level-3. If $G$ is of type 1, 2, 5, 6, 7, 8, 15, 16, or~17, $\N'_1$ and $\N'_2$ are isomorphic by \cref{lem:generator1-17}. If~$G$ is of type~14, the two networks are isomorphic modulo reversing the top side by \cref{lem:generator14}a. This proves Part~(a).

    (b) Suppose that $\Sub_{5} (\N_1 ) \simeq \Sub_{5} (\N_2)$, which implies $\Sub_{4} (\N_1 ) \simeq \Sub_{4} (\N_2)$. The proof is analogous to the proof of Part~(a) but using \cref{lem:generator14}b instead of \cref{lem:generator14}a when~$G$ is of type~14.
\end{proof}

\begin{remark}
If a network has fewer than four (respectively, five) leaves, then it possesses no quarnets (respectively, quinnets). Recall from the definition of encoding in \cref{sec:prelims} that, whenever a network has fewer than $k$ leaves, encoding by $k$-nets is understood to mean encoding by its $n$-nets, where $n$ is the number of leaves. Thus, in \cref{thm:main-encoding}, a network in $\mathcal{L}_3$ with fewer than five leaves is understood to be encoded by its unique subnetwork on all of its leaves.
\end{remark}

\section{Discussion}\label{sec:discussion}

We have shown that the class of binary, level-3, semi-directed networks is encoded by quinnets, and that its quarnets determine such a network exactly up to reversing the top side of its type~14 generators. In some definitions of semi-directed networks and their subnetworks, parallel edges and/or 2-blobs are kept intact (see, e.g., \cite{Allman2025a}). Without formalizing this setting, we note that our approach extends naturally to this framework and can be used to prove similar results in that framework too.
For the encoding and reconstruction problems we consider, maintaining parallel edges and 2-blobs actually simplifies the task of determining the generator of a blob. In particular, a leaf is no longer required to be attached to a side consisting of parallel edges in order to preserve those edges.
Similarly, a single subnetwork can encode information about many of the blobs in the remainder of the network, since the 2-blobs are not suppressed. Consequently, subnetworks carry strictly more information in this scenario.

It is also of interest to further formalize the connection to the rooted setting. Rooted level-2 networks are encoded by their trinets (3-nets)~\cite{vanIersel2014} and rooted level-3 networks by their quarnets~\cite{nipius2020rooted}, whereas level-2 and level-3 semi-directed networks need quarnets~\cite{huber2025encoding} and quinnets~(\cref{thm:main-encoding}), respectively. Hence, it appears that semi-directing costs one leaf when considering subnetwork encodings. A natural attempt at a formal reduction is via an \emph{outgroup} taxon: a rooted network on~$\X$ can be represented by a semi-directed network on $\X \cup \{o\}$, with the additional leaf~$o$ attached at the root. However, neither direction transfers immediately, since the $k$-nets of such a semi-directed network include subnetworks without~$o$, which carry no information about the root location, while conversely a semi-directed network need not have a leaf attached to its `source component'~\cite{holtgrefe2026characterizing} that could play the role of~$o$ in a reduction. Whether some reduction can nonetheless be made to work, or whether the two hierarchies are genuinely different would be interesting to
better understand.

Another direction is more algorithmic. For level-1, semi-directed networks, \cite{frohn2025reconstructing} showed that an $n$-leaved network can be reconstructed using only $\Oh(n \log n)$ quarnet queries, far fewer than the $\Theta(n^4)$ quarnets available. Nothing comparable is known beyond level-1, so it would already be interesting to know whether a level-2 network can be built from its quarnets in polynomial time, and, in view of \cref{thm:main-encoding}, whether the same can be done for a level-3 network from its quinnets, and how many queries are needed in either case.

Perhaps the most natural open question is whether encoding results of the kind proved here extend beyond level-3. The results of \cref{sec:encoding_features} hold for networks of arbitrary level, and it is only the final step---for which we rely on a case analysis over the level-3 generators---that is specific to level-3. Since the number of level-$k$ generators grows rapidly with~$k$~\cite{gambette2009structure,gambetteGenerators}, an argument along these lines is unlikely to settle the question in general, and a different approach appears to be needed.

Finally, we note that encoding is not the only obstacle to inferring a level-3 network in practice. Whether the $k$-nets themselves can be distinguished from one another using data is a separate question of statistical identifiability, an active area of research. This topic is quite well understood for level-1 networks but far less so for higher-level networks. The results available beyond level-1 are
restricted to particular classes of
networks~\cite{englander2025identifiability,Holtgrefe2025a,Allman2025a}, none of which cover level-3 networks in general. Almost all of this work, moreover, relies on
four-taxon summaries, with the corresponding theory for five-taxon summaries only beginning
to be developed~\cite{cummings2026identifiability}.

\paragraph{Data availability statement} 
The independent verification of the level-3 generator computations of~\cite{gambette2009structure} and~\cite{nipius2022encoding} in \cref{sec:encoding_level3} has been implemented in the open-source Python package \textsf{PhyloZoo}~\cite{holtgrefe2026phylozoo}.

\paragraph{Funding} NH and LvI were supported by grant OCENW.M.21.306 of the Dutch Research Council~(NWO). KTH and VM thank the TU Delft
Discrete Mathematics and Optimization group for hosting them during
visits where some of the work was carried out, and the aforementioned
NWO grant for helping to support their visits.

\printbibliography[title = References]

\end{document}